\documentclass{article}

\usepackage[T1]{fontenc}
\usepackage{booktabs,comment}  
\usepackage{fullpage}
\usepackage{enumitem}
\usepackage{xspace,xcolor}
\usepackage{multirow}
\usepackage{subcaption}
\usepackage{bbm,nicefrac}
\usepackage{tikz}
\usepackage{pgfplots}
\usepackage{subcaption}
\usepackage[square]{natbib}
\pgfplotsset{compat=1.18}
\usepackage{footnote}

\usepackage{amsmath,amsthm} 
\usepackage{thmtools,mathtools} 
\usepackage{thm-restate}
\usepackage{amsfonts,bm,dsfont} 
\usepackage[symbol]{footmisc}
\usepackage{placeins}
\usepackage{nicefrac,bbm}
\usepackage{subcaption}
\usepackage{appendix}
\usepackage[colorlinks=true,allcolors=black]{hyperref}

\usepackage[ruled]{algorithm2e} 

\SetAlFnt{\small}
\SetAlCapFnt{\small}
\SetAlCapNameFnt{\small}
\SetAlCapHSkip{0pt}
\IncMargin{-\parindent}

\usepackage{pgf,tikz,pgfplots}
\usetikzlibrary{matrix,calc,patterns,intersections,pgfplots.fillbetween}
\pgfplotsset{compat=1.18}

\newtheorem{theorem}{Theorem}[section]
\newtheorem*{theorem*}{Theorem}
\newtheorem{lemma}[theorem]{Lemma}
\newtheorem{proposition}[theorem]{Proposition}
\newtheorem{corollary}[theorem]{Corollary}
\newtheorem*{corollary*}{Corollary}
\newtheorem{claim}[theorem]{Claim}

\newtheorem{definition}[theorem]{Definition}

\DeclareMathOperator{\E}{\mathbb{E}}
\newcommand{\pr}{\mathrm{Pr}}
\allowdisplaybreaks

\newcommand{\R}{\mathbb{R}}

\newcommand{\opt}{\mathrm{OPT}}
\newcommand{\alg}{\text{ALG}}

\newcommand{\xm}{m}

\newcommand{\con}{C}
\newcommand{\rob}{R}

\newcommand{\talg}{\text{TAL}}
\newcommand{\talga}{\text{TAL}_{\alpha}}
\newcommand{\thr}{\talg}
\newcommand{\thra}{\talga}
\definecolor{DarkGreen}{RGB}{1,50,32}
\usepackage[]{color-edits}
\addauthor{JB}{blue!90!black}

\addauthor{IC}{blue!90!black}

\usepackage[nameinlink,capitalise]{cleveref}

\newcommand\blfootnote[1]{
  \begingroup
  \renewcommand\thefootnote{}
  \NoHyper\footnote{#1}\endNoHyper
  \addtocounter{footnote}{-1}
  \endgroup
}

\title{Prophet Inequality with Conservative Prediction}

\author{
Johannes Br\"ustle \\ Sapienza University of Rome
\and
Ilan Reuven Cohen \\ Bar Ilan University
\and
Stefano Leonardi \\ Sapienza University of Rome
}

\begin{document}

\maketitle

\blfootnote{\hspace*{-2.5em}
This project has been partially funded by the ERC Advanced Grant 788893 AMDROMA ``Algorithmic and Mechanism Design Research in Online Markets'', by the FAIR (Future Artificial Intelligence Research) project PE0000013, funded by the NextGenerationEU program within the PNRR-PE-AI scheme (M4C2, investment 1.3, line on Artificial Intelligence), and by the PNRR MUR project IR0000013-SoBigData.it, and by the MUR PRIN grant 2022EKNE5K (Learning in Markets and Society).
The work of I.R.~Cohen was supported by the Israel Science Foundation (grant No.~1737/21).}

\begin{abstract}
Prophet inequalities compare online stopping strategies against an omniscient "prophet" using distributional knowledge. In this work, we augment this model with a conservative prediction of the maximum realized value. We quantify the quality of this prediction using a parameter $\alpha \in [0,1]$, ranging from inaccurate to perfect. Our goal is to improve performance when predictions are accurate (consistency) while maintaining theoretical guarantees when they are not (robustness).

We propose a threshold-based strategy oblivious to $\alpha$ (i.e., with $\alpha$ unknown to the algorithm) that matches the classic competitive ratio of $1/2$ at $\alpha=0$ and improves smoothly to $3/4$ at $\alpha=1$. We further prove that simultaneously achieving better than $3/4$ at $\alpha=1$ while maintaining $1/2$
at $\alpha=0$ is impossible. Finally, when $\alpha$ is known in advance, we present a strategy achieving a tight competitive ratio of $\frac{1}{2-\alpha}$. 

\end{abstract}

\maketitle

\newpage

\section{Introduction}


Prophet inequalities are a pivotal tool in revenue management for analyzing robust, real-time decision-making under uncertainty. Originating in the optimal-stopping literature \citep{KrengelSucheston, samuel-cahn1984}, this framework compares the performance of online policies—which rely solely on distributional knowledge of future arrivals—against an optimal offline benchmark, the ``prophet,'' who possesses full foresight.

In the standard \emph{prophet inequality} setting, a decision-maker observes a sequence of non-negative random variables
\[
X_1, X_2, \ldots, X_n,
\]
which arrive sequentially. While the distributions are known a priori, the specific realizations are revealed one by one. Upon observing a variable, the decision-maker must strictly and irrevocably decide whether to accept it (terminating the process) or reject it (proceeding to the next). The objective is to design an online algorithm that maximizes the expected value of the selected element, providing a performance guarantee relative to the expected maximum
\[
\mathrm{OPT} = \mathbb{E}\!\left[\max_{i \in [n]} X_i\right],
\]
which represents the value achieved by a prophet who knows all realizations in advance.

Classical results by \citet{KrengelSucheston} and \citet{samuel-cahn1984} establish that threshold-based online algorithms can guarantee an expected reward of at least half the prophet's value. For instance, accepting the first variable exceeding the median of the distribution of the maximum yields this guarantee. A similar bound—optimal for general distributions—is achieved by the algorithm proposed by \citet{Kleinberg2012}, which selects the first realization exceeding half the expected maximum.

In practice, however, decision-makers often leverage machine-learned predictors or auxiliary information that estimate the maximum value more accurately than distributional data alone. We capture this scenario using a parsimonious prediction model. The online algorithm receives a conservative prediction $\tilde{v}$ approximating the specific realization of the prophet's value, $M(X) = \max_{i \in [n]} X_i$.  
We assume $\tilde{v}$ satisfies a multiplicative lower bound determined by a \emph{prediction quality parameter} $\alpha \in [0,1]$:
\[
\tilde{v} \;\ge\; \alpha \cdot M(X).
\]

Here, $\alpha \approx 1$ indicates a highly accurate lower bound, whereas $\alpha = 0$ reduces to the setting with no meaningful guarantees. A prediction comes from an external source like a machine learning algorithm that is able to predict a lower bound on the value that will be realized. Our main results are for the most interesting case in which the prediction is conservative, i.e., it is  not higher than the maximum realized value.  Learning conservative estimations of relevant quantities has recently been  studied in machine learning for reinforcement learning \cite{KZTL20, CYJ25, BGB21,LYS23}.  In statistics, several well known methods are available to provide a lower bound on the confidence interval of a random variable, (e.g., \cite{Hoeffding1963}) that will be exceeded by the realization with high probability.   On the other hand, we also study the case when the prediction fails to be lower than the realized maximum value. 

Our approach is closely related to the \emph{algorithms with predictions} framework \cite{MV22}, seeking methods that are simultaneously \emph{consistent} (enhancing performance when predictions are accurate) and \emph{robust} (maintaining reliability when predictions are inaccurate). 

Formally:
\begin{itemize}
    \item An algorithm is \emph{$\con$-consistent} if it is $\con$-competitive when the prediction is correct.
    \item An algorithm is \emph{$\rob$-robust} if it remains $\rob$-competitive even when the prediction is arbitrarily erroneous.
\end{itemize}

Our primary objective is to design an algorithm that remains $1/2$-robust when predictions fail (e.g., $\alpha =0$), while achieving $> 1/2$-consistency that improves smoothly as the prediction quality $\alpha$ increases. 
We demonstrate that this is achievable even when the algorithm is oblivious to the parameter $\alpha$. When $\alpha$ is known to the algorithm in advance, we find the exact level of consistency possible for all $\alpha \in [0,1]$.

\medskip

\subsection{Our Results and Techniques}

\paragraph{Thresholds for Prophets with Conservative Predictions}
 Our main results assume the conservative prediction model: the prediction $\tilde{v}$ is a lower bound on the maximum value that will be realized and always less than or equal to the realized maximum, thus satisfying $M(X) \ge \tilde{v} \ge \alpha \cdot M(X)$. 
 
 We adopt a prediction-augmented policy that sets an effective threshold $\max\{\tilde{v}, \tau\}$, where $\tau$ is a base threshold. However, the choice of the base threshold $\tau$ is non-trivial; standard threshold selections such as the median-based threshold \cite{samuel-cahn1984} or the mean-based threshold \cite{Kleinberg2012} fail to effectively leverage the prediction, as illustrated by the following examples.

\medskip

\begin{enumerate}
    \item \label{ex: median-ex} \textbf{(Median-based threshold)} Consider the standard threshold $\tau$ such that $\Pr[M(X) \ge \tau] = 1/2$. In an instance where $X_1=1$ deterministically and $X_2=2$ with probability $1/2$, the median threshold is $\tau=2$. This base threshold is too large: even with perfect predictions ($\alpha=1$), the algorithm is forced to reject $X_1$, yielding a competitive ratio of only $2/3$ while we will show that at least $3/4$ is achievable.
    \item \label{ex: mean-ex} \textbf{(Mean-based threshold)} Consider the threshold $\tau = \frac{1}{2}\mathbb{E}[M(X)]$. For an instance where $X_1=1$ and $X_2 = 1/\epsilon + 1/\sqrt{\epsilon}$ with probability $\epsilon$, $\E[M(X)] = 2+\sqrt{\varepsilon} - \varepsilon$, so that $\tau > 1$. Consequently, the algorithm always rejects the safe value $X_1$ in favor of the rare high value. In this case, the competitive ratio is limited to $1/2$ for any $\alpha$.
\end{enumerate}

\paragraph{Algorithms with Known Prediction Quality}
First, we consider the setting where the parameter $\alpha$ is known. Even in very simple instances, selecting the correct base
threshold is non-trivial once predictions are incorporated. As shown in Appendix~\ref{app:motivating_example}, the threshold achieving the desired competitive ratio may vary sharply with the distribution parameters, despite a fixed prediction quality $\alpha$.
This motivates a systematic approach aimed to determine the optimal way to incorporate  prediction into the threshold selection.  

\medskip

\noindent \textbf{Theorem (Competitive Ratio for Known $\alpha$)} For any instance where the prediction quality $\alpha\in [0,1]$ is known, we propose a threshold algorithm achieving a competitive ratio of $1/(2-\alpha)$. Moreover, no online algorithm can achieve a competitive ratio strictly better than ${1}/(2-\alpha)$ under $\alpha$-quality predictions.

\medskip

This result, depicted in \Cref{fig:known_alpha}, is significant as it establishes the fundamental limit of what can be achieved with conservative prediction in the setting where $\alpha$ is known; the bound is tight for every possible value of $\alpha$.

To show this, we define a novel base threshold $\tau^*_\alpha$, based on the threshold introduced by \citet{samuel1988prophet}.  
Since we obtain a tight approximation ratio, this demonstrates a surprising connection to \citet{samuel1988prophet}, who was seeking to strengthen a result by \citet{HillKertz1981} on prophet inequalities with the additional information of values being bounded. 

The intuition behind our threshold $\tau^*_\alpha$ lies in the structural properties of the worst-case instances. We show that the performance of the algorithm is convex in the probability of the largest item. This implies that the worst-case scenarios occur at the ``extremes'': either the largest item is effectively absent, or it is sufficiently likely to dictate the threshold. Our threshold is specifically calibrated to handle this latter ``direct case,'' effectively adapting the strategy of \citet{samuel1988prophet} for bounded variables to the prediction-augmented setting.

\paragraph{Algorithm Oblivious to Prediction Quality}
Second, we address the scenario where the algorithm is oblivious to $\alpha$. 

\medskip

\noindent \textbf{Theorem (Competitive Ratio of $\alpha$-oblivious Algorithm)}  We propose an $\alpha$-oblivious mechanism that guarantees a competitive ratio $f(\alpha)$ bounded by ${1}/{2} + \alpha^3/4 \leq f(\alpha) \leq {1}/{2} + \alpha/4$. Moreover, we show that no online algorithm can simultaneously maintain a competitive ratio of $1/2$ at $\alpha=0$ while achieving $3/4$ at $\alpha=1$.

\medskip

This result, depicted in \Cref{fig:oblivious_alpha}, constitutes a non-trivial generalization of the classic prophet inequality. Crucially, in the absence of actionable predictions (i.e., $\alpha=0$), our bound recovers the fundamental robustness guarantee of $1/2$. This confirms that our mechanism safely incorporates conservative predictions without incurring a penalty in the worst-case scenario; the decision-maker performs no worse than the standard baseline when predictions are uninformative, but gains significant value as quality improves.

In order to prove the result, we introduce the set of base thresholds $T_I(c)$ defined by the equation $\mathbb{E}[M \cdot \mathbb{I}[M \ge T_{I}(c)]] = c \cdot \opt$, which implicitly controls the product of the tail probability and the conditional expectation. This approach is conceptually distinct from the thresholds used thus far for standard prophet inequalities: whereas \citet{samuel-cahn1984} relies solely on a quantile (the median of the maximum) and \citet{Kleinberg2012} relies solely on the mean of the maximum, $T_I(c)$ integrates both metrics to secure a specific fraction of the optimal welfare. We set $c=3/4$ and use $\bar{\tau} =  T_I(3/4)$ as base threshold to optimize the trade-off for our purposes. 

A key property of this definition is that the resulting base threshold remains relatively small across all instances. As we observe in Examples \ref{ex: median-ex} and \ref{ex: mean-ex}, this is strategic: it ensures that the effective threshold, defined as the maximum of the base threshold and the prediction, remains sensitive to the prediction. When the prediction is accurate and high, it overrides the small base threshold, allowing the algorithm to fully leverage the power of the prediction. Nevertheless, when the prediction is low, the base threshold is sufficient to secure a robust welfare guarantee.

\begin{figure}[h]
    \centering
    \begin{subfigure}[b]{0.48\textwidth}
        \centering
        \begin{tikzpicture}
            \begin{axis}[
                width=\textwidth,
                height=6cm,
                axis lines=left,
                xlabel={$\alpha$},
                ylabel={Comp. Ratio},
                ymin=0.4, ymax=1.05,
                xmin=0, xmax=1,
                grid=major,
                title={\textbf{Known Prediction Quality}},
                clip=false
            ]
            \addplot[
                domain=0:1, 
                samples=100, 
                color=blue, 
                very thick
            ]
            {1/(2-x)};
            \addlegendentry{$\frac{1}{2-\alpha}$}
            
            \node[circle,fill,inner sep=1.5pt] at (axis cs:0,0.5) {};
            \node[anchor=south west, xshift=1pt, yshift=1pt] at (axis cs:0,0.5) {$1/2$};
            
            \node[circle,fill,inner sep=1.5pt] at (axis cs:1,1) {};
            \node[anchor=south east] at (axis cs:1,1) {$1$};
            \end{axis}
        \end{tikzpicture}
        \caption{The optimal competitive ratio, achieved by our algorithm, when $\alpha$ is known.}
        \label{fig:known_alpha}
    \end{subfigure}
    \hfill
    \begin{subfigure}[b]{0.48\textwidth}
        \centering
        \begin{tikzpicture}
            \begin{axis}[
                width=\textwidth,
                height=6cm,
                axis lines=left,
                xlabel={$\alpha$},
                ylabel={Comp. Ratio},
                ymin=0.4, ymax=1.05,
                xmin=0, xmax=1,
                grid=major,
                title={\textbf{Oblivious to Prediction Quality}},
                legend pos=north west,
                clip=false
            ]
            
            \addplot[
                domain=0:1, 
                samples=100, 
                color=red, 
                dashed, 
                thick
            ]
            {0.5 + x/4};
            \addlegendentry{Upper: $\frac{1}{2} + \frac{\alpha}{4}$}
            
            \addplot[
                domain=0:1, 
                samples=100, 
                color=blue, 
                very thick
            ]
            {0.5 + x^3/4};
            \addlegendentry{Lower: $\frac{1}{2} + \frac{\alpha^3}{4}$}
            
            \node[circle,fill,inner sep=1.5pt] at (axis cs:0,0.5) {};
            \node[anchor=south west, xshift=1pt, yshift=1pt] at (axis cs:0,0.5) {$1/2$};
            
            \node[circle,fill,inner sep=1.5pt] at (axis cs:1,0.75) {};
            \node[anchor=east, xshift=-3pt] at (axis cs:1,0.75) {$3/4$};
            
            \end{axis}
        \end{tikzpicture}
        \caption{Bounds on the competitive ratio of our algorithm when $\alpha$ is unknown.}
        \label{fig:oblivious_alpha}
    \end{subfigure}
    
    \caption{Comparison of competitive ratios for known vs. unknown prediction quality.}
    \label{fig:comparison}
\end{figure}
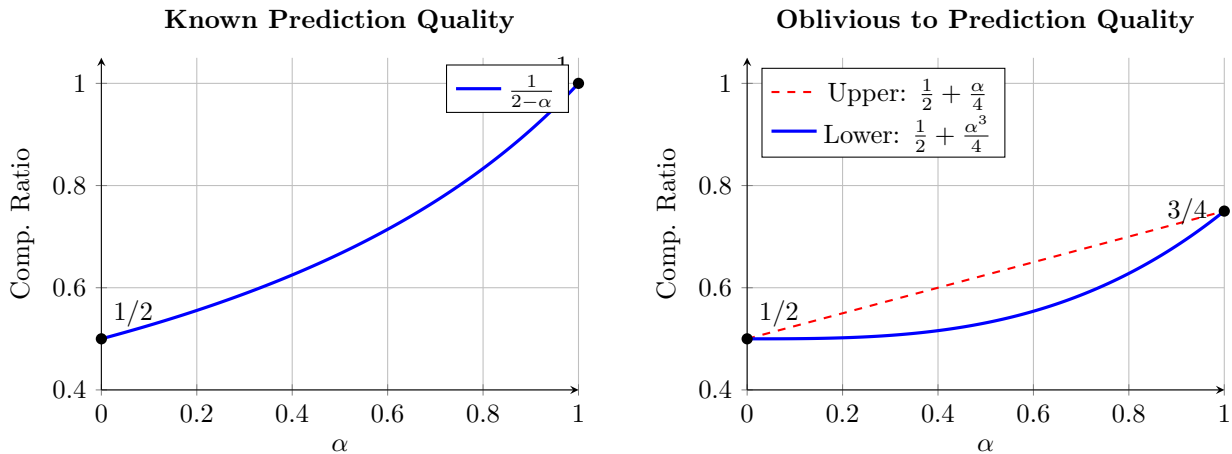

\paragraph{Non-conservative predictions.} We also consider the case where the prediction $\tilde{v}$ is \emph{non-conservative}, i.e., it may overestimate the true maximum. Suppose, however, that this overestimation is bounded by a multiplicative factor of $(1+\epsilon)$. In this setting, our scheme can still be applied by appropriately scaling down the prediction. Doing so yields slightly weaker guarantees, corresponding to a competitive ratio of $f\!\left(\frac{\alpha}{1+\epsilon}\right)$.

In Appendix~A, we further analyze the trade-off between robustness and consistency when the prediction $\tilde{v}$ may arbitrarily overestimate the true maximum. We show that even for a simple instance in this general setting, any algorithm that achieves consistency $\con$ and robustness $\rob$ must satisfy the bound $\con + \rob \le 1$. In other words, it is impossible to improve upon the classical worst-case guarantee without proportionally sacrificing performance when predictions are accurate.

\subsection{Related Work}
\citet{DT19} investigated this for the general prophet inequality problem, showing that in a sequence of $n$ values, a prior misprediction error $\epsilon$ (in Kolmogorov or L\'evy metric) degrades performance additively by $n \epsilon$. 
More recently, \citet{bai2025optimalstoppingpredictedprior} formulated the optimal stopping problem with a single predicted prior $\tilde{\mathcal{F}}$ versus an unknown true prior $\mathcal{F}$. Their work also focuses on the trade-off between consistency and robustness \cite{MV22}. For $n$ i.i.d.\ values, the secretary algorithm of \citet{dynkin} is both $1/e$-consistent and $1/e$-robust. Conversely, the optimal prophet inequality algorithm \cite{K86, HK82} is $0.745$-consistent but $0$-robust. The method in \citet{bai2025optimalstoppingpredictedprior} interpolates between these regimes, sacrificing the $1/e$-robustness of Dynkin’s algorithm to obtain the $0.745$ consistency bound.
In this setting, a recent work of \citet{KK25} shows how it is possible to simultaneously obtain $1/e$ robustness and $0.745$ consistency.  

Compared to previous works on prophet inequality with predictions referenced above, we remark that we adopt a different prediction model. Namely, we consider the general non-i.i.d. prophet inequality model and achieve perfect robustness, i.e., competitive ratio $1/2$, when the prediction is not accurate.  

 In the spirit of studying the theory of optimal stopping problems with predictions, we also mention works presenting several models of the secretary problem with prediction \cite{BS24, DLLV24, FY24, AGKK20}. In particular, \cite{AGKK20} is to some extent related to ours since it considers the secretary problem with inaccurate prediction (two-sided error/non-conservative) of the maximum value among all arriving secretaries. The results they obtain are $>1/e$ consistency when the prediction is accurate and $<1/e$ robustness when the prediction is not accurate. 
While their model of prediction is more general, we remark that their results do not extend to the general non i.i.d. prophet inequality.  Moreover, our results are stronger since we are able to prove $1/2$-robustness, corresponding to the tight lower bound obtained without prediction.

\section{Model and Definitions}

\paragraph{The Classic Prophet Inequality}

In the standard \emph{prophet inequality} setting, an online decision-maker observes a sequence of non-negative random variables with finite expectation
\[
X_1, X_2, \ldots, X_n,
\]
arriving one by one. The distributions of the variables are known in advance, but the realizations are revealed only sequentially. Upon observing $X_i$, the decision-maker must irrevocably decide whether to accept it, in which case the process stops, or to reject it and continue to the next variable.

Let 
\[
\mathrm{OPT} = \mathbb{E}\!\left[\max_{i \in [n]} X_i\right]
\]
denote the expected value obtained by a ``prophet'' who knows all realizations in advance.
The goal of the online algorithm is to maximize the expected value of the selected element,
denoted $\mathrm{ALG}$, and to provide a guarantee of the form
\[
\mathrm{ALG} \ge c \cdot \mathrm{OPT},
\]
for the largest possible constant $c \in (0,1)$.

A classical result of \cite{KrengelSucheston} and \cite{samuel-cahn1984} establishes that there exists a threshold-based online algorithm achieving the tight bound $c = 1/2$. 
\citet{Kleinberg2012} propose the algorithm that selects the first realization exceeding
\[
\tau = \frac{1}{2}\,\mathbb{E}\!\left[\max_{i \in [n]} X_i\right],
\]
guaranteeing that its expected reward is at least half of the prophet's value. This bound is optimal for general distributions.

\paragraph{Algorithms with Predictions}

We consider the prophet inequality problem in the \emph{algorithms with predictions} (or learning-augmented) paradigm. In addition to the distributions of the random variables $X_1, \ldots, X_n$, the online algorithm receives a conservative prediction $\tilde{v}$, intended to approximate the prophet's value $\max_{i \in [n]} X_i$.

We assume that the prediction satisfies a multiplicative lower bound (for general predictions, we have strong negative results; see \Cref{sec:genpred} for details). That is, they are given as
\[
\max_{i \in [n]} X_i \;\ge\; \tilde{v} \;\ge\; \alpha \cdot \max_{i \in [n]} X_i,
\]
where $\alpha \in [0,1]$ is a \emph{prediction quality parameter}. When $\alpha$ is close to $1$,
$\tilde{v}$ is guaranteed to be a reasonably accurate lower bound on the maximum, while $\alpha = 0$ corresponds to the fully adversarial case in which no meaningful quality guarantee
on $\tilde{v}$ is available.

Our goal is twofold:
\begin{enumerate}
    \item \textbf{Algorithms with known prediction quality.}
    When the prediction quality $\alpha$ is known to the algorithm, we seek to design online selection rules that optimally exploit this information. Formally, we aim to characterize the best achievable competitive ratio $f(\alpha)$ such that
    \[
    \mathbb{E}[\mathrm{ALG}(\tilde{v}, \alpha)]
    \;\ge\;
    f(\alpha) \cdot \mathrm{OPT}.
    \]

    \item \textbf{Robust algorithms without knowing prediction quality.}
    We also consider algorithms that do not know the prediction quality $\alpha$. The objective is to design a \emph{robust} algorithm that guarantees a nontrivial competitive ratio for all realizations of the prediction, while smoothly improving as $\alpha$ increases. In particular, when $\alpha = 0$ the algorithm should recover the classical prophet-inequality guarantee.
    Formally, we seek to maximize the competitive ratio $f(\alpha)$ such that 
    \[
    \mathbb{E}[\mathrm{ALG}(\tilde{v})]
    \;\ge\;
    f(\alpha) \cdot \mathrm{OPT}.
    \]
    For notational convenience, for a given algorithm $\alg$ and distribution $X$ we denote\\
    $\mathbb{E}_X[\mathrm{ALG}(\tilde{v})] = f_X(\alpha) \cdot \mathrm{OPT}$, implying $f(\alpha) = \inf_{X} f_X(\alpha)$.
\end{enumerate}

\paragraph{Prediction-Augmented Threshold Algorithms}

All of the algorithms we study are based on a natural extension of the classical threshold policy for prophet inequalities. The algorithm begins by selecting a \emph{base threshold} $\tau \ge 0$, which may depend on distributional information and, in some settings, on the prediction quality parameter $\alpha$.

Given a prediction $\tilde{v}$ satisfying $\tilde{v} \ge \alpha \cdot \max_{i} X_i$, we
define the \emph{effective threshold}
\[
\zeta \;=\; \max\{\tau, \tilde{v}\}.
\]
Intuitively, the prediction provides a lower bound on the prophet's value, and the algorithm adapts by never selecting an element below this predicted benchmark.

The prediction-augmented threshold algorithm proceeds as follows:
\begin{enumerate}
    \item Compute the effective threshold $\zeta = \max\{\tau, \tilde{v}\}$.
    \item Observe the sequence $X_1, X_2, \ldots, X_n$ online.
    \item Select the first index $i$ for which $X_i \ge \zeta$, and output $X_i$.
    \item If no such index exists, output $0$.
\end{enumerate}
Denote by $\talga(X,\tau)$ the expected value of the prediction-augmented threshold algorithm that posts threshold $\zeta$. 
Also, denote by $\talg(X,\tau)$ the expected value of the threshold algorithm that posts threshold $\tau$. We give the formal definition below.

\paragraph{Notation}

Given an instance $I = (X_1,\dots,X_n)$ and a realization $R(I) = (v_1,\dots,v_n)$, 
let $\xm(R) = \max_{i \in [n]} v_i$.

For a threshold $\tau \ge 0$, define
\[
i^*(R,\tau) = \arg\min \{\, i : v_i \ge \tau \,\},
\]
with the convention that $i^*(R,\tau)=\emptyset$ if no such index exists.
The realized threshold value is
\[
\thr(R,\tau) = 
\begin{cases}
v_{i^*(R,\tau)}, & i^*(R,\tau)\neq \emptyset, \\
0, & \text{otherwise}.
\end{cases}
\]
Accordingly, the expected thresholded value is
\[
\thr(I,\tau) = \mathbb{E}_{R \sim I}[\thr(R,\tau)].
\]

Given a prediction quality parameter $\alpha \in [0,1]$, define the 
\emph{$\alpha$–augmented threshold}
\[
\thra(R,\tau) = \thr\!\left(R,\ \max\{\tau,\, \alpha \cdot \xm(R)\}\right),
\qquad
\thra(I,\tau) = \mathbb{E}_{R \sim I}[\thra(R,\tau)].
\]

\begin{claim}
\label{claim:threshold_equivalence}
For any instance $I$ and any prediction $\tilde v$ satisfying
$\alpha \cdot \max_{i} X_i \le \tilde v \le \max_{i} X_i$,
the quantity $\thra(I,\tau)$ is a lower bound on the expected value of the
prediction--augmented threshold algorithm with base threshold $\tau$.
\end{claim}

\begin{proof}
Fix a realization $R$ and let $v^* = \alpha \cdot \xm(R)$ denote the 
worst-case consistent prediction.  
Consider any feasible prediction value $v' \in [v^*,\, \xm(R)]$.  
Let $\zeta = \max\{\tau, v^*\}$ be the effective threshold induced by $v^*$, and let $\zeta' = \max\{\tau, v'\}$ be the effective threshold induced by $v'$.

If $i^*(R,\zeta') = i^*(R,\zeta)$, then $\thr(R,\zeta') = \thr(R,\zeta)$ and the algorithm’s realized value is unchanged.

If $i^*(R,\zeta') \neq i^*(R,\zeta)$, then by definition $\zeta < \zeta' \le \xm(R)$.
Thus $i^*(R,\zeta')$ is well-defined and
\[
v_{i^*(R,\zeta')} \;\ge\; \zeta' > v_{i^*(R,\zeta)}.
\]
Hence the realized value under threshold $\zeta'$ is weakly larger than under $\zeta$.

Therefore, for every realization $R$, the value obtained with any feasible prediction 
$v' \in [v^*,\,\xm(R)]$ is at least $\thra(R,\tau)$.  
Taking expectation over realizations yields the claim.
\end{proof}

\section{Algorithm for Known \texorpdfstring{$\alpha$}{alpha}}
\label{sec:known_pred}

In this section, we study the setting in which the prediction quality parameter
$\alpha$ is known to the algorithm in advance.
Our goal is to characterize the optimal competitive ratio achievable as a function
of $\alpha$, and to design a prediction-augmented threshold algorithm that attains
this bound. Note that, a naive algorithm which simply set the prediction as the threshold is only $\alpha$-competitive.

Let $I=(X_1,\ldots,X_n)$ be a sequence of non-negative independent random variables.
For a fixed $\alpha \in [0,1]$, we show that the optimal competitive ratio is
\(
f(\alpha) \;=\; \frac{1}{2-\alpha}.
\)
In particular, there exists a suitable choice of base threshold $\tau$ such that
the corresponding prediction-augmented threshold algorithm achieves a competitive
ratio of at least $f(\alpha)$, and no online algorithm can do better.

\medskip
\noindent
Our main result is stated below.

\begin{theorem}
\label{thm:mainknown}
Fix $\alpha \in (0,1)$. For any instance $I$ of non-negative independent random variables,
there exists a base threshold $\tau$ such that
\[
\thra(I,\tau) \;\ge\; \frac{1}{2-\alpha}\,\opt(I).
\]
\end{theorem}

We also show that this guarantee is tight: even when $\alpha$ is known, no online
algorithm can achieve a strictly better competitive ratio under $\alpha$-quality
predictions.

\begin{lemma}
\label{lem:known_alpha_lb}
For every $\alpha \in [0,1]$, no online algorithm can achieve a competitive ratio
strictly greater than $\frac{1}{2-\alpha}$ under $\alpha$-quality predictions.
\end{lemma}

The proof of Lemma~\ref{lem:known_alpha_lb} appears in \Cref{app:deferred_proofs}.

\subsection*{Reduction to Canonical Instances}

We begin by showing that it suffices to analyze instances of a particularly simple
form.

\begin{definition}[Sorted scaled Bernoulli instance]
An instance $I=(X_1,\ldots,X_n)$ is a \emph{sorted scaled Bernoulli instance} if
\begin{enumerate}[label=(\roman*)]
    \item each $X_i$ is supported on $\{0,v_i\}$ for some $v_i \ge 0$, and
    \item the values are ordered as $0 \le v_1 \le v_2 \le \cdots \le v_n$.
\end{enumerate}
\end{definition}

The following lemma shows that restricting attention to this class of instances
incurs no loss of generality.

\begin{lemma}
\label{lem:transform}
Fix $\alpha \in [0,1]$ and a function $f(\alpha)$.
If for every sorted scaled Bernoulli instance $I$ there exists a base threshold
$\tau$ such that
\[
\thra(I,\tau) \;\ge\; f(\alpha)\,\opt(I),
\]
then the same guarantee holds for every instance of non-negative independent
random variables.
\end{lemma}

The proof of Lemma~\ref{lem:transform}, which closely follows the approach of
\citet{ekbatani2024prophet}, is deferred to the appendix.

\subsection*{Threshold Construction for Bernoulli Instances}

It therefore remains to establish the desired competitive ratio for sorted scaled
Bernoulli instances.

\begin{lemma}
\label{lem:mainknownexist}
Fix $\alpha \in (0,1)$ and let $I$ be a sorted scaled Bernoulli instance.
Then there exists a base threshold $\tau$ such that
\[
\thra(I,\tau) \;\ge\; \frac{1}{2-\alpha}\,\opt(I).
\]
\end{lemma}

The proof of Lemma~\ref{lem:mainknownexist} constitutes the main technical contribution
of this section and is presented below.
Together, Lemmas~\ref{lem:transform} and~\ref{lem:mainknownexist} imply
Theorem~\ref{thm:mainknown}.

\medskip
\noindent
From this point onward, we assume without loss of generality that the instance $I$
is a sorted scaled Bernoulli instance. Our threshold selection builds on the \emph{SC-threshold}, introduced by \citet{samuel1988prophet}.
For an instance $I$, the \emph{SC-threshold} is defined as
\(
\eta(I) \;=\; \sup\{\tau \ge 0 \mid \thr(I,\tau) \ge \tau\}.
\)

In \citet{samuel1988prophet}, it is shown that when each random variable is supported on $[0,1]$, the following inequality holds:
\begin{equation}
\label{eq:ester}
\opt(I)
\;\le\;
2\,\thr(I,\eta(I)) - \thr(I,\eta(I))^2.
\end{equation}

For a sorted scaled Bernoulli instance $I$ with maximum value $v_n$, if $\eta(I) \ge \alpha \cdot v_n$,
then using the base threshold $\eta(I)$ the prediction-augmented algorithm coincides with the
standard threshold algorithm. In this case, as we show later, Equation~\eqref{eq:ester} implies
\(
\thra(I,\eta(I)) = \thr(I,\eta(I))
\;\ge\;
\frac{1}{2-\alpha}\,\opt(I),
\)
as required.
If this condition fails, we show that it suffices to restrict attention to a suitable prefix
$I^{(k)}$ of the instance, where $I^{(k)}$ denotes the sub-instance consisting of the first
$k$ variables $(X_1,\ldots,X_k)$. 

We now define the base threshold used in the proof of
Lemma~\ref{lem:mainknownexist}.
Specifically, we select the SC-threshold of the largest prefix $I^{(k)}$ for which
$\eta(I^{(k)}) \ge \alpha v_k$.

Formally, let
\[
s_\alpha = \max\bigl\{\, k \in [n] \mid \eta(I^{(k)}) \ge \alpha v_k \bigr\},
\qquad
\tau^*_\alpha = \eta(I^{(s_\alpha)}).
\]

\begin{lemma}
\label{lem:mainknown}
Fix $\alpha \in (0,1)$ and let $I$ be a sorted scaled Bernoulli instance. Then
\[
\thra(I,\tau^*_\alpha) \;\ge\; \frac{1}{2-\alpha}\,\opt(I).
\]
\end{lemma}

First, we establish several basic properties of the SC-threshold and of the base
threshold $\tau^*_\alpha$ that will be used throughout the proof. Let $I=(X_1,\ldots,X_n)$ be a sorted scaled Bernoulli instance, where each random variable
$X_i$ is supported on $\{0,v_i\}$ and
\[
\Pr[X_i=v_i]=p_i,
\qquad
\Pr[X_i=0]=1-p_i.
\]
 Then the prophet’s expected value satisfies the standard recursion
\begin{equation}
\label{eq:opt}
\opt(I)
= v_n p_n + (1-p_n)\opt(I^{(n-1)}).
\end{equation}

Next, fix a threshold $\tau \in (v_{r-1},\,v_r]$. Since the instance is sorted, the only
values that can meet or exceed $\tau$ are $v_r,\ldots,v_n$, and the threshold algorithm
accepts the first index whose realization is at least~$\tau$. Therefore, posting threshold
$\tau$ is equivalent to posting threshold $v_r$, and we obtain the recursion
\begin{equation}
\label{eq:thr}
\thr(I,\tau)=\thr(I,v_r)
= p_r\,v_r + (1-p_r)\,\thr(I,v_{r+1}).
\end{equation}

\begin{claim}
For every $k \in \{1,\ldots,n-1\}$,
\begin{equation}
\label{eq:thrdiff}
\thr(I,v_k)-\thr(I,v_{k+1})
\;=\;
p_k\bigl(v_k-\thr(I,v_{k+1})\bigr).
\end{equation}
\end{claim}

\begin{proof}
Applying the recursion \eqref{eq:thr} with threshold $v_k$ gives
\[
\thr(I,v_k)=p_k v_k+(1-p_k)\thr(I,v_{k+1}).
\]
Subtracting $\thr(I,v_{k+1})$ from both sides yields
\[
\thr(I,v_k)-\thr(I,v_{k+1})
= p_k v_k - p_k \thr(I,v_{k+1})
= p_k\bigl(v_k-\thr(I,v_{k+1})\bigr),
\]
which proves the claim.
\end{proof}

For a sorted scaled Bernoulli instance, the maximum value is $v_n$, and therefore we can
apply the above bound after rescaling by $v_n$. This yields the following corollary.

\begin{corollary}\label{cor:ester-bound}
\[
\frac{\opt(I)}{v_n}
\;\le\;
2\cdot \frac{\thr(I,\eta(I))}{v_n}
-\left(\frac{\thr(I,\eta(I))}{v_n}\right)^2.
\]
\end{corollary}

We first prove two helpful claims. The following one proves a monotonicity property for the standard prophet inequality. It shows that once the base threshold exceeds the SC-threshold $\eta(I)$, increasing it further to the next available value $v_{k+1}$ cannot increase the expected reward. This confirms that the SC-threshold effectively identifies a ``peak'' in performance for the standard algorithm, as selecting strictly higher values yields diminishing returns.

\begin{claim}
\label{cl:basester}
For any $k \in [n]$, if \(v_k > \eta(I)\), then  
\(
\thr(I,v_k) \ \ge\ \thr(I,v_{k+1}).
\)
\end{claim}

\begin{proof}
The definition of \(\eta(I)\) implies that  
\[
\thr(I,\tau) < \tau \qquad \text{for all } \tau > \eta(I).
\]
In particular, because \(v_{k+1} \ge v_k > \eta(I)\), we obtain  
\[
\thr(I,v_{k+1}) < v_{k+1}.
\]

Suppose for contradiction that \(\thr(I,v_{k+1}) = v_k + \epsilon\) for some \(\epsilon > 0\).  
Then since $v_k+\epsilon \in (v_k,v_{k+1}]$, we would have  
\[
\thr(I,v_k+\epsilon) = \thr(I,v_{k+1}) = v_k + \epsilon,
\]
which contradicts the fact that \(v_k + \epsilon > \eta(I)\) and hence  
\(\thr(I,v_{k+1}) \leq v_{k}\).

Now using identity~\eqref{eq:thrdiff}, we have  
\[
\thr(I,v_k) - \thr(I,v_{k+1})
    = p_k \bigl(v_k - \thr(I,v_{k+1})\bigr).
\]
Since \(\thr(I,v_{k+1}) \le v_k\), the right-hand side is nonnegative, and thus  
\[
\thr(I,v_k) \ge \thr(I,v_{k+1}),
\]
as claimed.
\end{proof}

The next claim extends the monotonicity result to the prediction-augmented setting. It demonstrates that for any base threshold above our chosen $\tau^*_\alpha$, the expected performance of the algorithm with predictions does not increase. We show that when the realized maximum is large enough to trigger the prediction, both thresholds result in the same effective threshold. When the prediction is not used, the logic reduces to the standard monotonicity established in Claim \ref{cl:basester}.

Recall that we denote $M(X) = \max\{X_1,\ldots,X_n\}$, or $M$ when the dependence on $X$ is implicit. 
For event $A$
we denote  $\thra(I,\tau |  A) = \mathbb{E}_{R \sim I}[\thra(R,\tau) | A]$.

\begin{claim}
\label{cl:mono_pred}
For any $k \in [n]$, if $v_k > \tau^*_\alpha$, then 
\[
\thra(I,v_k) \;\ge\; \thra(I,v_{k+1}).
\]
\end{claim}

\begin{proof}
First, observe that
\begin{equation}
    \thra(I, v_k \mid M > v_k/\alpha)
    \;=\;
    \thra(I, v_{k+1} \mid M > v_k/\alpha),
\end{equation}
since for all realizations with \(M > v_k/\alpha\) the two base thresholds induce the same effective threshold.

Indeed, if \(M \ge v_{k+1}/\alpha\), then in both cases \(\zeta = \alpha M\).
Otherwise, if
\[
M \in \left( \frac{v_k}{\alpha},\, \frac{v_{k+1}}{\alpha} \right],
\]
then under base threshold \(v_k\) we have \(\zeta = \alpha M \in (v_k,\, v_{k+1}]\), which is equivalent to using threshold \(v_{k+1}\).
Thus, in all cases with \(M > v_k/\alpha\), the algorithm behaves identically under base thresholds \(v_k\) and \(v_{k+1}\), and the conditional expectations coincide.

Second, let $t = \arg \max_i \{v_i \leq v_k/\alpha\}$. We note that $\thra(I,v_{k} | M \leq v_k/\alpha) = \thr(I^{(t)}, v_k)$
and 
$\thra(I,v_{k+1} | M \leq v_k/\alpha) = \thr(I^{(t)}, v_{k+1})$, that is the prediction doesn't have an effect for base thresholds $v_k$ and $v_{k+1}$.
Next, we will show that $v_k >   \eta(I^{(t)})$.
 If $\eta(I^{(t)})\geq v_k$, since we know $v_k \geq v_t \cdot \alpha$ by definition of $t$, we get that $\eta(I^{(t)})\geq v_t \cdot \alpha$ which implies $s_\alpha\geq t$ by definition of $s_\alpha$. Thus we obtain
$$ v_k > \tau^*_\alpha= \eta(I^{(s_\alpha)}) \geq \eta(I^{(t)}) \geq v_k,$$
 a contradiction. Thus we must have $v_k >   \eta(I^{(t)})$.

Therefore, $v_k > \eta(I^{(t)})$ and by Claim~\ref{cl:basester} we have 
\begin{equation}
    \thra(I,v_{k} | M \leq v_k/\alpha) = \thr(I^{(t)}, v_k)
    \geq \thr(I^{(t)},v_{k+1}) = \thra(I,v_{k+1} | M \leq v_k/\alpha)
\end{equation} 

By putting this together, we have
\begin{align*}
    \thra(I,v_k) &=  \thra(I,v_{k} | M > v_k/\alpha) \cdot \pr[ M > v_k/\alpha] + \thra(I,v_{k} | M \leq v_k/\alpha) \cdot \pr[ M \leq v_k/\alpha]
    \\ 
    &=  \thra(I,v_{k+1} | M > v_k/\alpha) \cdot \pr[ M > v_k/\alpha] + \thra(I,v_{k} | M \leq v_k/\alpha) \cdot \pr[ M \leq v_k/\alpha]\\
    &\geq  \thra(I,v_{k+1} | M > v_k/\alpha) \cdot \pr[ M > v_k/\alpha] + \thra(I,v_{k+1} | M \leq v_k/\alpha) \cdot \pr[ M \leq v_k/\alpha]\\
&=  \thra(I,v_{k+1})
\end{align*}

\end{proof}

\begin{proof}[Proof of Lemma~\ref{lem:mainknown}]
Consider a sorted, scaled Bernoulli instance $I$.  
If $\eta(I) \ge \alpha \cdot v_n$, then $\tau^*_\alpha = \eta(I)$, and by Corollary~\ref{cor:ester-bound},
\[
\frac{\thra(I,\tau^*_\alpha)}{\opt(I)}
    \;\ge\;
    \frac{\tau^*_\alpha}{2\tau^*_\alpha - (\tau^*_\alpha)^2 / v_n}
    \;=\;
    \frac{1}{2 - \tau^*_\alpha/v_n}
    \;\ge\;
    \frac{1}{2 - \alpha}.
\]

Otherwise, define a modified instance $I(x)$ in which we replace 
$X_n = (v_n,p_n)$ with $X_n' = (v_n, x)$.  
Note that both $\thra(I(x),\tau)$ and $\opt(I(x))$ are linear functions of $x$.

Let $p_n^*$ be the value of $x$ such that $\eta(I(p_n^*)) = \alpha \cdot v_n$.  

Since $\tau^*_\alpha \leq \eta(I(p_n)) = \eta(I) < \alpha \cdot v_n$, it follows that $p_n \in (0, p_n^*)$.

Because both $\thra(I(x),\tau)$ and $\opt(I(x))$ are affine functions of $x$,
the ratio $\frac{\thra(I(x),\tau)}{\opt(I(x))}$ attains its minimum at one of the
endpoints of the interval $[0,p_n^*]$. Indeed, the ratio of two affine functions with a
positive denominator is monotone on the interval, since its derivative has a
constant sign. Therefore,
\[
\frac{\thra(I(p_n),\tau^*_\alpha)}{\opt(I(p_n))}
\;\ge\;
\min\left\{
    \frac{\thra(I(0),\tau^*_\alpha)}{\opt(I(0))},
    \;
    \frac{\thra(I(p_n^*),\tau^*_\alpha)}{\opt(I(p_n^*))}
\right\}.
\]

If
\[
\frac{\thra(I(0),\tau^*_\alpha)}{\opt(I(0))}
    \;\le\;
    \frac{\thra(I(p_n^*), \tau^*_\alpha)}{\opt(I(p_n^*))},
\]
then removing variable $n$ (i.e., setting $x=0$) yields the smaller instance $I(0)$, and we continue recursively on that instance.

Let $t = \arg \min_i \{v_i \geq \alpha \cdot v_n\}$,
by definition $\thra(I(p_n^*),v_t) = \alpha \cdot v_n$
and $\frac{\thra(I(p_n^*),v_t)}{\opt(I(p_n^*))} = \frac{\thr(I(p_n^*),v_t)}{\opt(I(p_n^*))} \geq \frac{1}{2-\alpha}$ by Corollary ~\ref{cor:ester-bound}.  

We conclude the proof by showing
\[ \thra(I(p_n^*),\tau^*_\alpha) \geq \thra(I(p_n^*),v_t). \]

First, we have
\[ \thra(I(p_n^*),\tau^*_\alpha | M = v_n) = \thr(I(p_n^*),v_t | M = v_n) = \thra(I(p_n^*),v_t | M = v_n)\]

and $\thra(I(p_n^*),\tau^*_\alpha | M < v_n) = \thra(I^{(n-1)},\tau^*_\alpha) $
and  $\thra(I(p_n^*),v_t | M < v_n) =  \thra(I^{(n-1)},v_t)$.
Note that, $\tau^*_\alpha < v_t$ is the chosen threshold also for $I^{(n-1)}$ and therefore by Claim~\ref{cl:mono_pred}, we have 
$$\thra(I^{(n-1)},\tau^*_\alpha) \geq \thra(I^{(n-1)},v_t)$$.
\end{proof}

\section{\texorpdfstring{$\alpha$}{alpha}-Oblivious Algorithms}\label{sec:threshold oblivious}

In this section, unless stated otherwise, we assume for convenience that the instance $X = (X_1,\dots,X_n)$ is a tuple of continuous random variables. At the end of \Cref{sec:extension to general}, we comment on how a modification of our algorithm can be implemented when distributions contain point masses. We start by defining our mechanism by proposing a novel base threshold in Section \ref{sec:approx_mechanism} and state our main theorem, which gives its approximation guarantees. In the following sections we prove these approximation bounds by first addressing the case without prediction ($\alpha=0$) followed by the case $\alpha \in (0,1]$. In Section \ref{sec:oblivious lower bound} we show that these bounds are tight at the extremes, that is $\alpha \in \{0,1\}$.

\subsection{Approximation Mechanism}\label{sec:approx_mechanism}

First we define the family of base thresholds among which we will choose a suitable candidate. Recall that we denote $M(X) = \max\{X_1,\ldots,X_n\}$, or $M$ when the dependence on $X$ is implicit. 
\begin{definition}
For an instance $X = (X_1,\dots,X_n)$ and a constant $c \in (0,1)$, denote by $T_X(c) \in \R_+$ the unique threshold such that
\begin{equation}\label{def:oblivious threshold}
\E[M\cdot \mathbbm{1} [M \geq T_X(c)]] = c\cdot \opt.   
\end{equation}
\end{definition}
Throughout Section \ref{sec:threshold oblivious}, we interchangeably denote $T_c = T_X(c)$ when $X$ is implied from context, and in particular denote $\bar{\tau} = T_X(3/4)$.

Our main theorem shows that setting $c = 3/4$ is indeed a good choice: we analyze the prediction augmented threshold algorithm $\talga(X,\bar{\tau})$. Notably, the lower bound on its competitive ratio is strictly increasing in $\alpha$. In particular, for $\alpha = 0$ and $\alpha = 1$, we recover $f(0) \geq 1/2$ and $f(1) \geq 3/4$, which in \Cref{sec:oblivious lower bound} we show is a tight result.
Crucially, we remark that $\talga(X,\bar{\tau})$ can be implemented \emph{without} knowledge of $\alpha$.

\begin{theorem}\label{thm:oblivious_main}
Let $\alpha \in [0,1]$. Given any instance $X$, consider the prediction quality oblivious algorithm with prediction: $\talga(X,\bar{\tau})$. Then its competitive ratio is bounded by 
\[
\left(\frac{1}{2}+\frac{\alpha^3}{4}\right) \leq f(\alpha) \leq \left(\frac{1}{2}+\frac{\alpha}{4}\right).\]
\end{theorem}

The main part of this section consists of proving the lower bound on $f(\alpha)$. 

We start by analyzing our $\alpha$-augmented threshold algorithm in the case of $\alpha=0$ or equivalently the setting without prediction in Section \ref{sec:alpha = 0}. Next, we consider $\alpha >0$. Our approach is to separately analyze the cases where the threshold is equal to the base threshold in Section \ref{sec:base threshold}, and where the threshold is equal to the prediction in Section \ref{sec:prediction_threshold}. Finally, we carefully combine the two results and prove the main theorem in Section \ref{sec:oblivious main theorem}.

\subsection{Robustness when \texorpdfstring{$\alpha = 0$}{alpha = 0}}\label{sec:alpha = 0}

We prove a stronger claim than what is strictly necessary for the proof of the main theorem when $\alpha = 0$, and parametrize the bound on $\talg(X,\bar{\tau})$ by the probability $\Pr[M \ge \bar{\tau}]$.

\begin{lemma}\label{lem:alpha=0}
Consider an instance $X$ such that $p = \Pr[M(X) \ge \bar{\tau}] \in (0,1)$. Then 
\begin{equation}\label{eq:alpha_0 ratio}
   \talg(X,\bar{\tau}) \ge \left(\frac{4p^2 - 6p + 3}{4(1-p)}\right)\cdot{\opt}.
\end{equation}
This ratio is exactly $1/2$ if and only if $p = 1/2$, and is strictly greater than $1/2$ for all other $p \in (0, 1)$.
\end{lemma}

\begin{proof}
Since $X_i$ are continuous for $i\in [n]$, so is $M$. Denote by $F_M$ its cumulative distribution function and by $f_M$ its probability density function. 

By section 2 in \cite{CorreaSurvey2019}, we know that
\begin{equation}\label{eq:first alpha_0 ineq}
\talg(X,\bar{\tau}) \geq p\cdot \bar{\tau} + (1-p)\int_{\bar{\tau}}^\infty (1-F_{M}(x))dx.
\end{equation}

We expand the second term and apply the definition of $T$ to obtain
\begin{align}\label{eq:max_above_T}
    \begin{split}
    & \int_{\bar{\tau}}^{\infty} (1-F_M)dx = \int_{\bar{\tau}}^\infty \int_{x}^\infty f_{M}(t)dt dx = \int_{\bar{\tau}}^\infty \int_{\bar{\tau}}^t f_{M}(t) dx dt \\
    & = \int_{\bar{\tau}}^\infty xf_{M}(t) |_{\bar{\tau}}^t dt = \int_{\bar{\tau}}^\infty(t-\bar{\tau})f_{M}(t)dt = \frac{3}{4}\E[M] - \bar{\tau}\cdot \Pr[M > \bar{\tau}].
    \end{split}
\end{align}

By combining (\ref{eq:first alpha_0 ineq}) and (\ref{eq:max_above_T}), we get
\begin{equation}\label{eq:2nd lower bound alpha_0}
\talg(X,\bar{\tau}) \geq  (1-p)\cdot(3/4)\cdot \E[M] + p^2 \cdot \bar{\tau}.
\end{equation}

We express $\opt$ using the law of total expectation:
\begin{equation*}
    \opt = \E[M] = \mathbb{E}[M \mid M \ge \bar{\tau}]p + \mathbb{E}[M \mid M < \bar{\tau}](1-p).
\end{equation*}
By the definition of $\bar{\tau}$, we substitute the first term with $\frac{3}{4}\mathbb{E}[M]$:
\begin{equation*}
    \mathbb{E}[M] = \frac{3}{4}\mathbb{E}[M] + \mathbb{E}[M \mid M < \bar{\tau}](1-p) \quad \Leftrightarrow  \quad \frac{1}{4}\mathbb{E}[M] = \mathbb{E}[M \mid M < \bar{\tau}](1-p).
\end{equation*}

Since $\mathbb{E}[M \mid M < \bar{\tau}] \le \bar{\tau}$, we have:
\begin{equation*}
    \frac{1}{4}\mathbb{E}[M] \le \bar{\tau}(1-p) \implies \bar{\tau} \ge \frac{\mathbb{E}[M]}{4(1-p)}.
\end{equation*}

Substituting this lower bound for $\bar{\tau}$ into inequality (\ref{eq:2nd lower bound alpha_0}):
\begin{align*}
    \talg(X,\bar{\tau}) &\ge (1-p)\frac{3}{4}\mathbb{E}[M] + p^2 \left( \frac{\mathbb{E}[M]}{4(1-p)} \right) = \mathbb{E}[M] \left( \frac{3(1-p)}{4} + \frac{p^2}{4(1-p)} \right).
\end{align*}
We combine the terms over a common denominator $4(1-p)$:
\[
    \left( \frac{3(1-p)}{4} + \frac{p^2}{4(1-p)} \right) = \frac{3(1-p)^2 + p^2}{4(1-p)} = \frac{4p^2 - 6p + 3}{4(1-p)},
\]
which proves the first part of the claim, (\ref{eq:alpha_0 ratio}).
The remainder of the claim is implied by showing that
$g(p) = (4p^2 - 6p + 3)/(4(1-p))$
is strictly convex and takes its minimum at $p=1/2$. Moreover, $g(1/2)=1/2$. Lemma \ref{lem:g(p) func} shows these statements and thereby concludes the proof.
\end{proof}

\begin{lemma}\label{lem:g(p) func}
Let $g: (0,1) \to \mathbb{R}$ be defined by
$
g(p) = (4p^2 - 6p + 3)/({4(1-p)}).
$
The function $g$ is strictly convex on the interval $(0,1)$ and satisfies the inequality $g(p) \ge \frac{1}{2}$ for all $p \in (0,1)$.
\end{lemma}

\subsection{Base Threshold Bound}\label{sec:base threshold}

\medskip

In this section, we give a tight bound for the special case where $X_i$ are all bounded exactly such that the prediction is never used. Moreover, our claim holds more generally for any base threshold $T_X(c)$ where $c\in [0,1]$.

\begin{lemma}\label{lem:smallvalues_noprediction}
For a constant $c \in [0,1]$, and $\alpha\in (0,1]$, if $X_i \in [0, T_X(c)/\alpha]$ for all $i \in [n]$, then $\talga(X,T_X(c)) = \talg(X,T_X(c))$ obtains
\[
f_X(\alpha) \ge c \left( \frac{2 + (\alpha - 2)c + 2\sqrt{1-c}}{2 - c + 2\sqrt{1-c}} \right).
\]
\end{lemma}

\begin{proof}
For simplicity of notation, though we will define several instances, denote the base threshold of the original instance by $T_c = T_X(c)$. 
The proof will proceed by taking a number of steps to construct an instance $H$ that satisfies several properties.
\begin{enumerate}
    \item $T_X(c) = T_H(c)$
    \item $\E[M(X)] = \E[M(H)]$
    \item $\talg(X,T_c) \geq \talg(H,T_c)$
    \item $H$ is a sequence of Bernoulli distributions, increasing in weight.
\end{enumerate}
First, we define $Y_i$ similarly as in \citep{samuel1988prophet}, that is for $i\in [n]$,
\[
Y_i = 
\begin{cases}
(1/\alpha)\cdot T_c &\text{ w.p } p_i = \frac{1}{1-T_c} \int_{T_c}^{(1/\alpha)T_c} (x - T_c) dF_{X_i}\\
T_c &\text{ w.p } q_i = \frac{1}{1-T_c} \int_{T_c}^{(1/\alpha)T_c} (1 - x) dF_{X_i}\\
0 &\text{ otherwise.}
\end{cases}
\]
where $F_{X_i}$ is the cumulative distribution function of $X_i$.

By section 2 in \citep{samuel1988prophet}, this construction ensures that the conditional expectations and probabilities are preserved, namely 
\[\mathbb{E}[X_i | X_i \ge T_c] = \mathbb{E}[Y_i | Y_i \ge T_c]\text{ and }\Pr[X_i \ge T_c] = \Pr[Y_i \ge T_c] \quad \forall i\in [n].\]
From these properties, it immediately follows that  
\begin{equation}
\talg(X,T_X(c))  = \talg(Y,T_X(c))  \label{eq:samuelcahn eq}.
\end{equation}
However, we don't have $T_X(c) = T_Y(c)$. We fix this by further transforming $Y$ in two steps, first by addressing values greater or equal to $T$. By Lemma 2.1 in \citep{samuel1988prophet}, $\mathbb{E}[M(X) | M(X) \geq T_c] \le \mathbb{E}[M(Y) | M(Y) \geq T_c]$.
Thus we may transform $Y$ once more to an instance $Z$ such that 
$\mathbb{E}[M(X) | M(X) \geq T_c] = \mathbb{E}[M(Z) | M(Z) \geq T_c]$.
One simple way of doing this is to note there exists a unique $\beta_1 \leq 1$ such that if we define
\[
Z_i = 
\begin{cases}
(1/\alpha)\cdot T_c &\text{ w.p } \beta_1 p_i \\
T_c &\text{ w.p } q_i \\
0 &\text{ otherwise}
\end{cases}
\]

then 
\begin{equation}
\mathbb{E}[M(X) | M(X) \geq T_c] = \mathbb{E}[M(Z) | M(Z) \geq T_c].\label{eq:XZ_equality}
\end{equation}

Clearly this transform decreases the algorithm value as well, hence by (\ref{eq:samuelcahn eq}), we have
\begin{equation}
\talg(X,T_c) \geq \talg(Z,T_c) \label{eq:XZ_alg_ineq}.
\end{equation}
Second, we need to add back to $Z$ some weight on values strictly smaller than $T_c$.
\begin{equation*}
 \Pr[X_i \geq T_c] \geq \Pr[Z_i \geq T_c] \quad \forall i \in [n] 
\Rightarrow \Pr[M(X) \geq T_c] \geq \Pr[M(Z) \geq T_c].   
\end{equation*}
Define $\beta_2\leq 1$ by
\[
\beta_2 = \frac{\Pr[M(X) < T_c]}{\Pr[M(Z) < T_c]}. 
\]
We now aim to define a variable $Z_0$ so that $T_Z(c)=T_c$.

By continuity, there exists some $\varepsilon>0$ such that $\E[M(X) | M(X) < T_c] \leq T_c-\varepsilon$. 
Define
\[
Z_0 = 
T_c-\varepsilon \text{ w.p } z = \frac{\beta_2\cdot\E[M(X) | M(X) < T_c]}{T_c-\varepsilon}.
\]
and update notation $Z = \{Z_0,Z_1,\dots,Z_n\}$.

Since $\beta_2\leq 1$ and $\E[M(X) | M(X) < T_c] \leq T_c- \varepsilon$, $Z_0$ is well defined. 

Then we get 
\begin{align*}
& \E[M(Z)|M(Z)<T_c]\cdot \Pr[M(Z)<T_c] \\
= &  (T_c-\varepsilon)\cdot\frac{\beta_2\cdot\E[M(X) | M(X) < T_c]}{T_c-\varepsilon} \cdot (1/\beta_2)\Pr[M(X) <T_c] \\
= & \E[M(X) | M(X) < T_c] \cdot \Pr[M(X) <T_c]
\end{align*}
and conclude that $T_Z(c) = T_c$.
Moreover, by (\ref{eq:XZ_equality}) and (\ref{eq:XZ_alg_ineq}), it follows that 
\begin{equation}\label{eq:XtoZ}
    f_X(\alpha) \geq f_Z(\alpha).
\end{equation}

We further reduce the instance by a transformation  into Bernoulli random variables. 
For $i \in [n]$ let
\[
V_i^1 = T_c \ \quad w.p \quad q_i \quad \text{and} \quad V_i^2 = T_c/\alpha \ \quad w.p \quad \beta_1p_i/(1-q_i).
\]

This is well defined since if $\beta_1p_i/(1-q_i)>1 \Leftrightarrow q_i + \beta_1p_i >1$, contradicting that $Z_i$ is well defined. 
We compute $\talg(V,T_c)$. The probability of stopping at $V_i^1$ or $V_i^2$ given that the algorithm has not stopped at any index $<i$ is: $q_i + (1-q_i)\cdot (\beta_1pi/(1-q_i)) = q_i + \beta_1p_i$. Similarly, the expected value of the algorithm at $V_i^1$ and $V_i^2$ conditioned on reaching $i$ is $q_i\cdot T_c + \beta_1p_i \cdot (T_c/\alpha)$. 
Following these calculations, we consider the threshold algorithm on the instance 
\[
V = \{Z_0,V_1^1, V_1^2, V_2^1, V_2^2,\ldots, V_n^1, V_n^2\}.
\]
We get
\begin{equation}
    \talg(Z,T_c)  = \talg(V,T_c).
\end{equation}
Also, it is clear from the definition of $V$ and $Z$ that
\begin{equation}
    \E[M(Z)] \leq \E[M(V)].
\end{equation}
As before, we instead force the expected maxima to remain the same and the value of the algorithm to decrease. 
Indeed, there exists $\beta_3 \leq 1$ such that if for $i \in [n]$ we  let
\[
W_i^1 = T_c \ \quad w.p \quad a_i = q_i \quad \text{and} \quad W_i^2 = T_c/\alpha \ \quad w.p \quad b_i = \beta_3\beta_1p_i/(1-q_i) 
\]
and define 
\[
W = \{Z_0,W_1^1, W_1^2, W_2^1, W_2^2,\ldots, W_n^1, W_n^2\}
\]
then we get
\begin{equation}
\talg(Z,T_c)  \geq \talg(W,T_c) \quad \text{and} \quad  \E[M(Z)] = \E[M(W)].
\end{equation}  
This also implies $T_W(c) = T_c$ and hence 
\begin{equation}\label{eq:ZtoW}
f_Z(\alpha) \geq f_W(\alpha).
\end{equation}

We now reduce the instance to only $3$ Bernoulli variables. 
Let
\[
H_1 = T_c  \quad w.p \quad r_1 = 1- \prod_{i=1}^n (1-a_i) \quad \text{and} \quad H_2 = T_c/\alpha  \quad w.p \quad r_2 = 1- \prod_{i=1}^n (1-b_i) 
\]
and denote $H = \{Z_0,H_1,H_2\}$. 
Then 
\[\E[M(H)|M \geq T_c]\cdot \Pr[M(H)\geq T_c] = r_2\cdot T_c/\alpha + r_1(1-r_2)\cdot T_c = \E[M(W)|M \geq T_c]\cdot \Pr[M(W)\geq T_c].
\]
Thus $T_H(c) = T_c$ and we can compute
\[\talg(H,T_c) = \talg(Z_0,W_1^1,\ldots,W_n^1,W_1^2,\dots,W_n^2,T_c) \leq \talg(W,T_c)\]
which allows us to conclude  
\[f_H(\alpha)\leq f_W(\alpha)\]
and hence $f_H(\alpha) \leq f_X(\alpha)$ by the previous reductions, (\ref{eq:XtoZ}) and (\ref{eq:ZtoW}).

We aim to obtain a lower bound on $f_H(\alpha)$. For fixed $T_c$ and $\alpha$, we may formulate this as an optimization problem. Given any $3$ Bernoulli random variables of weights $T_c-\varepsilon, T_c, T_c/\alpha$ with probabilities $x_0,x_1,x_2$ respectively, in order for the definition of the threshold to be satisfied, we must have
\[
T_c/\alpha \cdot x_2 + T_c\cdot(1-x_2)x_1 = (c/(1-c))\cdot(T_c-\varepsilon)\cdot(1-x_2)(1-x_1)x_0.
\]
The value of the threshold algorithm is given by
\[
T_c\cdot x_1 + (T_c/\alpha)\cdot x_2(1-x_1).
\]
Since we let $\varepsilon \rightarrow 0$, we may omit it, as well as normalize $T_c=1$.
This corresponds to the optimization problem of Proposition \ref{prop:reduced_optimization} with $k = 1/\alpha$ and $t=c/(1-c)$, which proves our claim. 
\end{proof}

The following proposition is the optimization problem that corresponds to analyzing the worst case of $f_H(\alpha)$ over all simple bounded instances $H$ of 3 weighted Bernoulli random variables, as defined precisely in the reduction of Lemma \ref{lem:smallvalues_noprediction}.

\begin{proposition}\label{prop:reduced_optimization}
Let $k>1$ and $t>0$ be constants. For $x_0, x_1, x_2 \in [0,1]$, the minimum value of the function
\[ f(x_0, x_1, x_2) = \frac{x_1 + k x_2 (1-x_1)}{k x_2 + (1-x_2)(1-(1-x_1)(1-x_0))} \]
subject to the constraint
\[ t \cdot x_0(1-x_1)(1-x_2) = k x_2 + (1-x_2)x_1 \]
is
\[ f_{min} = \frac{t}{t+1} \left( \frac{2k + t + 2k\sqrt{t+1}}{k(2+t+2\sqrt{t+1})} \right). \]
\end{proposition}

Our application of Lemma \ref{lem:smallvalues_noprediction} will actually not be for precisely $c=3/4$, but for some unknown $c\in [0,3/4]$. Therefore, we need to understand the behavior of the approximation ratio obtained, in function of $c$.

\begin{lemma}\label{lem:smallval fn decreasing}
For any fixed constant $\alpha \in (0, 1)$, the function $g:[0, 1] \to \mathbb{R}$ defined by
\[ g(c) = \frac{2 + (\alpha - 2)c + 2\sqrt{1-c}}{2 - c + 2\sqrt{1-c}} \]
is strictly decreasing on its domain.
\end{lemma}

Lastly, we prove the upper bound part of Theorem \ref{thm:oblivious_main}, which is based on the the worst case $n=3$ instance analyzed in \Cref{prop:reduced_optimization}. We defer its proof to \Cref{app:deferred_proofs}.

\begin{lemma}\label{lem:tight_oblivious_instance} 
For any $\alpha \in [0,1]$ and $\varepsilon>0$, there exists an instance $X$ such that
\[\talga(X,\bar{\tau}) \leq \left(\left(\frac{1}{2}+\frac{\alpha}{4}\right) + \varepsilon\right)\cdot \opt.\]
\end{lemma}

\subsection{Prediction Threshold Bound}\label{sec:prediction_threshold}

Now, we address the performance of our threshold algorithm when $M(X) > \bar{\tau}/\alpha$, that is, when the effective threshold is set to be the prediction. Since the behavior of the threshold is complicated in this case, our analysis is based on bounding $\Pr[M(X) > \bar{\tau}/\alpha]$ in Lemma \ref{lem:T/alpha probability upper bound}, whose proof we defer to \Cref{app:deferred_proofs}.
Such a bound becomes useful along with the observation that for any realization of $X$, unless there exists a second value in $[M(X)\cdot \alpha, M(X)]$ other than $M(X)$, the algorithm obtains the maximum value. In Lemma \ref{lem:>T/alpha values}, we leverage this fact to get a bound on the expected value of $\talga(X,\bar{\tau})$ in the event that $M(X)$ is sufficiently large.

\begin{lemma}\label{lem:T/alpha probability upper bound}
\[
\Pr[M(X) \geq \bar{\tau}/\alpha] \leq  \frac{1}{1/(3\alpha) +1}.
\]
\end{lemma}

Combining the cases of small and large values of $M(X)$ is non-trivial, an additional approximation proves to be useful and we apply Lemma \ref{lem:T/alpha probability upper bound} not on $M(X) > \bar{\tau}/\alpha$, but $M(X) > \bar{\tau}/\alpha^3$. This gives the following bound on the expected value of our threshold algorithm for large maximum values. 

\begin{lemma}\label{lem:>T/alpha values}
Let $B$ be the event that the prediction is used, i.e., $B = \{ M(X)> \bar{\tau}/\alpha^3 \}$. Then we get
\[
 \E_{R \sim X}[\talg_{\alpha}(R\cdot\mathbbm{1}_B,\bar{\tau})] \ge ((1-g(\alpha)) + \alpha g(\alpha))\cdot \mathbb{E}[M(X) \cdot \mathbf{1}_{B}]
\]
where $g(x) = \frac{1}{1/(3x^2)+1}$.
\end{lemma}

\begin{proof}
For simplicity denote $M =M(X)$ and let $Y$ be the value obtained by the algorithm. The event of interest is $B = \{ M > \bar{\tau}/\alpha^3 \}$. In this event, the algorithm's threshold is $M\cdot \alpha$. In this case, the algorithm's value is always at least its threshold, so $Y \ge M\cdot \alpha$.

The algorithm's value is in fact exactly $M$ unless it is "stopped early" by a value $v_j \sim X_j$ that arrives before the maximum item $M=v_k \sim X_k$, where $v_j \ge M\cdot \alpha$ and $j<k$. Let's condition on the outcome of the maximum value, $M=v_k$, where $v_k>\bar{\tau}/\alpha^3$.  
Let $S_{<k}$ be the event that the algorithm stopped before $k$, given that $M = v_k$.

We give a simple bound on the expected value, 
\[ \mathbb{E}[Y | M=v_k] \ge v_k (1 - \Pr[S_{<k}]) + (v_k \cdot \alpha)\Pr[S_{<k}] = v_k - v_k (1-\alpha)\cdot\Pr[S_{<k}] \]
Let $S = \{i \mid v_i \ge \bar{\tau}/\alpha^2\}$. For a given maximum $M=v_k$, the stopping probability $\Pr[S_{<k}]$ is at most the probability that at least one item in $S \setminus \{k\}$ is realized. Let this event be $E_{k}$. That is, $\Pr[S_{<k}]  \leq \Pr[E_{k}]$. Note how $E_k$ measures the conditional performance of the second order statistic. In order to bound $\Pr[E_k]$ we give the following Lemma and defer its proof to \Cref{app:deferred_proofs}. 

\begin{lemma}\label{lem:order_stats_bound}
Let $X_1, \dots, X_n$ be independent random variables. Let $M_1$ and $M_2$ denote the largest and second largest values among the $X_i$, respectively (i.e., the $n$-th and $(n-1)$-th order statistics). For any fixed thresholds $L_1$ and $L_2$ such that $L_1 > L_2$, the following inequality holds:
\[
\Pr[M_2 > L_2 \mid M_1 > L_1] \leq \Pr[M_1 > L_2].
\]
\end{lemma}

Thus combining \Cref{lem:order_stats_bound} with \Cref{lem:T/alpha probability upper bound}, we obtain $\Pr[E_{k}] \le g(\alpha)$ where $g(x) = 1/((1/3x^2)+1)$. 
Averaging over all possible outcomes for the maximum value $M$ within the event $B$, we have:
\[ \mathbb{E}[Y \cdot \mathbf{1}_B] = \mathbb{E}_M\left[ \mathbb{E}_Y[Y|M] \cdot \mathbf{1}_B \right] \ge \mathbb{E}\left[ (M - {M}(1-\alpha)\Pr[S_{<k}]) \cdot \mathbf{1}_B \right] \]
Using the bound $\Pr[S_{<k}]  \le g(\alpha)$:
\[ \mathbb{E}[Y \cdot \mathbf{1}_B] \ge \mathbb{E}\left[ (M - (1-\alpha){M} \cdot g(\alpha)) \cdot \mathbf{1}_B \right] =\left((1-g(\alpha)) + \alpha g(\alpha)\right)\mathbb{E}[M \cdot \mathbf{1}_B] \]
\end{proof}

Since the bound we obtained from Lemma \ref{lem:>T/alpha values} has a complicated form, we give a simpler, monotone lower bound.

\begin{lemma}\label{lem:compare_ratios}
Let $g(x) = \frac{1}{1/(3x^2)+1}$. For all $x \in [0,1]$, the following inequality holds:
\[ (1-g(x))+ xg(x) \ge \frac{2+x^3}{3}. \]
\end{lemma}

\subsection{Proof of the Main Theorem}\label{sec:oblivious main theorem}

Finally, we are ready to combine the results of the previous section in order to prove the main result on our robust mechanism.
The proof establishes the global lower bound by partitioning the outcome space into two disjoint events determined by the magnitude of the maximum realization $M(X)$. We distinguish between a ``base threshold regime'' (denoted as event $A$), where $M(X)$ falls within the interval $[\bar{\tau}, \bar{\tau}/\alpha^3]$, and a ``prediction regime'' (denoted as event $B$), where $M(X)$ exceeds $\bar{\tau}/\alpha^3$. By the results from Section \ref{sec:base threshold} and \ref{sec:prediction_threshold}, we know that the algorithm achieves a conditional competitive ratio of at least $(2+\alpha^3)/3$ in both regimes. We conclude by using the definition of $\bar{\tau}=T_X(3/4)$ once again, which implies that $3/4$ of welfare is obtained in $A \cup B$, that is
$\E[M(X) \cdot \mathbbm{1}_{A\cup B}] = (3/4)\cdot \E[M(X)]$.

\begin{proof}[Proof of Theorem \ref{thm:oblivious_main}]
The upper bound holds by Corollary \ref{lem:tight_oblivious_instance}. We proceed to prove the lower bound. 
For $\alpha = 0$, the claim holds by Lemma \ref{lem:alpha=0}.
Now let $\alpha \in (0,1]$.
Define the events $A= \{X: \bar{\tau} \leq M(X)\leq \bar{\tau}/\alpha^3\}$ and $B = \{X: M(X) >  \bar{\tau}/\alpha^3\}$. Also, let
\[
Y_i = 
\begin{cases}
    &  X_i \quad \text{when} \quad X_i \leq \bar{\tau}/\alpha^3\\
    & 0 \quad \text{otherwise}.
\end{cases}
\]
Then
\[
\E[M(X)\mathbbm{1}_A] = \E[M(Y)\mathbbm{1}_A], \quad \E_{R \sim X}[\talg_{\alpha}(R\cdot\mathbbm{1}_A,\bar{\tau})] = \E_{R \sim Y}[\talg_{\alpha}(R\cdot\mathbbm{1}_A,\bar{\tau})]   
\]
\[
\text{ and } \E[M(X)\mathbbm{1}[M(X) < \bar{\tau}]] = \E[M(Y)\mathbbm{1}[M(Y) < \bar{\tau}]].
\]
However, we have that $\bar{\tau} = T_Y(c)$ for some $c\leq 3/4$, with the inequality possibly being strict.
Since certainly $\talg_{\alpha}(Y,\bar{\tau}) \geq \talg_{\alpha^3}(Y,\bar{\tau})$, we may apply Lemma \ref{lem:smallvalues_noprediction} to obtain
\[
\frac{\talg_{\alpha}(Y,\bar{\tau})}{\mathbb{E}[M(Y)]} \ge \frac{\talg_{\alpha^3}(Y,\bar{\tau})}{\mathbb{E}[M(Y)]} \ge c \left( \frac{2 + (\alpha^3 - 2)c + 2\sqrt{1-c}}{2 - c + 2\sqrt{1-c}} \right).
\]
Since by definition $c\cdot \mathbb{E}[M(Y)] = \mathbb{E}[M(Y)\mathbbm{1}_A]$, we get 

\[
\E_{R \sim Y}[\talg_{\alpha}(R\cdot \mathbbm{1}_A,\bar{\tau})] \geq 
\left(\frac{2 + (\alpha^3 - 2)c + 2\sqrt{1-c}}{2 - c + 2\sqrt{1-c}} \right)\cdot \mathbb{E}[M(Y)\mathbbm{1}_A].
\]
We further bound this by noting that $c\leq 3/4$, hence by Lemma \ref{lem:smallval fn decreasing}, we know the minimum is taken at $c = 3/4$. Indeed, setting $c$ at this value,
\[
\left(\frac{2 + (\alpha^3 - 2)c + 2\sqrt{1-c}}{2 - c + 2\sqrt{1-c}} \right) = \frac{2+\alpha^3}{3}
\]
implying 
\begin{equation}\label{eq:interval A ineq}
    \E_{R \sim Y}[\talg_{\alpha}(R\cdot\mathbbm{1}_A,\bar{\tau})] \geq 
\left(\frac{2+\alpha^3}{3} \right)\cdot \mathbb{E}[M(Y)\mathbbm{1}_A].
\end{equation}

By Lemma \ref{lem:>T/alpha values}, we directly get
\[
\E_{R \sim X}[\talg_{\alpha}(R\cdot\mathbbm{1}_B,\bar{\tau})] \ge ((1-g(\alpha)) + \alpha g(\alpha)) \cdot {\mathbb{E}[M(X) \cdot \mathbf{1}_{B}]}
\]
and by Lemma \ref{lem:compare_ratios} we know that 
$(1-g(\alpha)) + \alpha g(\alpha) \geq (2+\alpha^3)/3$ so that we have
\begin{equation}\label{eq:interval B ineq}
\E_{R \sim X}[\talg_{\alpha}(R\cdot\mathbbm{1}_B,\bar{\tau})]  \ge \left(\frac{2+\alpha^3}{3} \right) \cdot {\mathbb{E}[M(X) \cdot \mathbf{1}_{B}]}
\end{equation}

Finally, by (\ref{eq:interval A ineq}) and (\ref{eq:interval B ineq}),

\begin{align*}
\frac{\talg_{\alpha}(X,\bar{\tau})]}{\E[M(X)]}
= & \left(\frac{3}{4}\right)\frac{ \E_{R \sim X}[\talg_{\alpha}(R\cdot\mathbbm{1}_A,\bar{\tau})] + \E_{R \sim X}[\talg_{\alpha}(R\cdot\mathbbm{1}_B,\bar{\tau})]}{\E[M(X)\mathbbm{1}_A]+ \E[M(X)\mathbbm{1}_B]}\\
= & \left(\frac{3}{4}\right)\frac{\E_{R \sim Y}[\talg_{\alpha}(R\cdot\mathbbm{1}_A,\bar{\tau})] + \E_{R \sim X}[\talg_{\alpha}(R\cdot\mathbbm{1}_B,\bar{\tau})]}{ \mathbb{E}[M(Y)\mathbbm{1}_A]+ \E[M(X)\mathbbm{1}_B]}\\
\geq & \left(\frac{3}{4}\right)\left(\frac{2+\alpha^3}{3} \right) = \frac{1}{2}+\frac{\alpha^3}{4}.
\end{align*}
\end{proof}

\subsection{Lower bound for \texorpdfstring {$\alpha=0,1$} {alpha}}\label{sec:oblivious lower bound}
We show that no (possibly randomized) algorithm that does not know $\alpha$ in advance can simultaneously guarantee
$f(0) \ge 0.5$ and $f(1) \ge 0.75 + \epsilon$ for any $\epsilon>0$.

\begin{lemma}\label{lem:unknown-alpha-lb}
For every $\epsilon>0$, there exists an instance and an adversarial choice of advice
such that any algorithm that does not know $\alpha$ cannot achieve both
$f(0)\ge 0.5$ and $f(1)\ge 0.75+\epsilon$.
\end{lemma}

\begin{proof}
We consider a two-variable instance.
Let $X_1 \sim B(1,1)$ and
$X_2 \sim B\!\left(\frac{1}{\epsilon}+\frac{1}{\sqrt{\epsilon}},\epsilon\right)$.
Note that
$E[X_2] = 1 + \sqrt{\epsilon}.$ The adversary chooses the advice parameter $\alpha \in \{0,1\}$ as follows:
\begin{enumerate}
    \item For $\alpha = 1$, the adversary sets $v^* = X_2$ when $X_2>0$, and
          $v^* = 1$ otherwise.
    \item For $\alpha = 0$, the adversary sets $v^* = 1$.
\end{enumerate}

Let $p$ denote the probability that the algorithm selects the first item when $v^*=1$.
Without loss of generality, we assume that the algorithm selects the second item whenever
$v^*>1$.
We denote by $\alg(\alpha,p)$ the expected value obtained by the algorithm with advice
parameter $\alpha$, where $p$ is the probability of selecting the first item in the case
$v^*=1$.

The optimal expected value is
\[
\opt
= \left(\frac{1}{\epsilon}+\frac{1}{\sqrt{\epsilon}}\right)\epsilon
  + 1\cdot(1-\epsilon)
= 2+\sqrt{\epsilon}-\epsilon.
\]

The algorithm’s expected reward in the two scenarios is given by
\begin{align*}
\alg(0,p)
&= p\cdot 1 + (1-p)\cdot E[X_2]
= p + (1-p)(1+\sqrt{\epsilon}),\\
\alg(1,p)
&= \left(\frac{1}{\epsilon}+\frac{1}{\sqrt{\epsilon}}\right)\epsilon
  + 1\cdot(1-\epsilon)\cdot p.
\end{align*}

We now distinguish between two cases.

\paragraph{Case 1: $p \ge 0.5 + \sqrt{\epsilon}$.}
Since $\alg(0,p)$ is monotone decreasing in $p$, we obtain
\[
\alg(0,p) \le \frac{2-2\epsilon+\sqrt{\epsilon}}{2} \quad \Rightarrow \quad \frac{\alg(0,p)}{\opt}
\le
\frac{\frac{\sqrt{\epsilon}-2\epsilon+2}{2}}
     {2+\sqrt{\epsilon}-\epsilon}
\le \frac{1}{2}.
\]

\paragraph{Case 2: $p < 0.5 + \sqrt{\epsilon}$.}
In this case,
\begin{align*}
 \frac{\alg(1,p)}{\opt}
 &\le
 \frac{-\epsilon^{3/2} + 2\sqrt{\epsilon} - \tfrac12\epsilon + \tfrac32}
      {(2-\sqrt{\epsilon})(\sqrt{\epsilon}+1)} \le
 \frac{3}{4} + \sqrt{\epsilon}.
\end{align*}

Thus, for any choice of $p$, either
$\frac{\alg(0,p)}{\opt} < 0.5$
or
$\frac{\alg(1,p)}{\opt} < 0.75 + \sqrt{\epsilon}$.
In particular, it is impossible to guarantee both
$f(0)\ge 0.5$ and $f(1)\ge 0.75+\epsilon$ for any $\epsilon>0$.
\end{proof}

\bibliographystyle{ACM-Reference-Format}
\bibliography{bibliography}

\appendix
\section{General Predictions}
\label{sec:genpred}
In this section we study a fully general prediction model, in which the algorithm
receives a prediction for each individual random variable.
Formally, in addition to observing the distributions of
$X_1,\ldots,X_n$, the online algorithm is given a vector of predictions
$\tilde{v} = (\tilde v_1,\ldots,\tilde v_n)$, where each $\tilde v_i \in \mathbb{R}_+$
is an arbitrary prediction for the realized value of $X_i$.
No structural assumptions are imposed on the prediction vector:
the predictions may overestimate or underestimate the realized values, and may be
chosen adversarially.
The algorithm may use $\tilde{v}$ arbitrarily when making its online decisions.

We show that under such general predictions the consistency--robustness tradeoff is
essentially trivial.
Specifically, any algorithm that is $\con$-consistent and $\rob$-robust must
satisfy $\con+\rob \le 1$.
Conversely, a simple algorithm that interpolates between trusting the prediction and
executing a worst-case optimal strategy achieves this bound for any $\con$.
We begin by establishing the impossibility result.

\begin{lemma}\label{lem:general-upper-bound}
Every algorithm that is $\con$-consistent and $\rob$-robust
under general predictions satisfies
\[
\con+\rob \le 1.
\]
\end{lemma}

\begin{proof}
Consider the classical two-item prophet inequality instance.
Item~$1$ has a deterministic value equal to~$1$.
Item~$2$ takes value $1/\epsilon$ with probability~$\epsilon$ and value~$0$ otherwise.

The expected offline optimum is
\[
\opt
= \epsilon \cdot \frac{1}{\epsilon} + (1-\epsilon)\cdot 1
= 2-\epsilon.
\]

Assume that the algorithm is given a prediction $\tilde{v}$ of the item values.
By definition of the instance, $\tilde{v}(1)=1$ and
$\tilde{v}(2)\in\{0,1/\epsilon\}$.

Any prediction-based policy can be parameterized by two probabilities $p,q\in[0,1]$:
\begin{itemize}
    \item If $\tilde{v}(2)=0$, the policy selects item~$1$ with probability~$p$.
    \item If $\tilde{v}(2)=1/\epsilon$, the policy selects item~$2$ with probability~$q$.
\end{itemize}

\paragraph{Consistency.}
When the prediction is correct, the expected reward of the policy is
\begin{align*}
&(1-\epsilon)\cdot p \cdot 1
+ (1-\epsilon)\cdot (1-p)\cdot 0
\\ &\qquad
+ \epsilon \cdot q \cdot \frac{1}{\epsilon}
+ \epsilon \cdot (1-q)\cdot 1
\\ &=
p + q + \epsilon \cdot (1-p-q).
\end{align*}

\paragraph{Robustness.}
When the prediction is adversarially chosen, the expected reward becomes
\begin{align*}
&(1-\epsilon)\cdot q \cdot 0
+ (1-\epsilon)\cdot (1-q)\cdot \frac{1}{\epsilon}
\\ &\qquad
+ \epsilon \cdot p \cdot 1
+ \epsilon \cdot (1-p)\cdot \frac{1}{\epsilon}
\\ &=
2 - \bigl(p + q + \epsilon(1-p-q)\bigr).
\end{align*}

Normalizing by $\opt$, the consistency and robustness guarantees satisfy
\[
\con + \rob
\le
\frac{2}{2-\epsilon}.
\]
Since $\epsilon>0$ is arbitrary, letting $\epsilon\to 0$ yields
$\con+\rob \le 1$.
\end{proof}

We now show that this bound is tight.

\begin{lemma}\label{lem:general-achievability}
There exists an algorithm such that for any $\rob \in [0, \tfrac{1}{2}]$,
the algorithm is $\con$-consistent and $\rob$-robust, with
$\con+\rob = 1$.
\end{lemma}

\begin{proof}
Consider the following strategy.
With probability~$p$, the algorithm follows a \emph{prediction-based policy},
i.e., a policy that selects the item with the highest predicted value.
With probability~$1-p$, it ignores the prediction and executes a worst-case
$c$-competitive algorithm.

If the prediction is correct, the prediction-based policy obtains value~$\opt$
with probability~$p$, and with probability~$1-p$ the algorithm obtains reward
at least $c\cdot \opt$.
Hence, the expected reward is at least
$p\cdot \opt + (1-p)\cdot c\cdot \opt$,
yielding a consistency guarantee of
\[
\con = p + (1-p)c.
\]

If the prediction is adversarial, the prediction-based policy may obtain zero reward.
However, with probability~$1-p$ the algorithm ignores the prediction and executes the worst-case $c$-competitive algorithm, obtaining expected reward at least $(1-p)c\cdot \opt$.
This yields a robustness guarantee of
\[
\rob = (1-p)c.
\]

Therefore,
\[
\con+\rob = p + 2c(1-p).
\]
For $c=\tfrac12$ (e.g., the classical threshold algorithm), this gives
\[
\con+\rob = 1.
\]

Finally, for any desired $\rob \in [0, \tfrac12]$, one can choose
$p = 1 - 2\rob \in [0,1]$, which yields the corresponding robustness guarantee $\con = 1-\rob$ and establishes $\con+\rob=1$.
\end{proof}

\section{Motivating example for known $\alpha$.}
\label{app:motivating_example}
We illustrate that, even in very simple instances, selecting the appropriate base threshold is non-trivial once predictions are incorporated.

Consider an instance consisting of three independent \emph{sorted scaled Bernoulli} random variables with prediction quality parameter $\alpha = 1/2$:
\[
X_1 = (1,1), \qquad X_2 = (2,p_2), \qquad X_3 = (4,p_3),
\]
where $X_i = (v_i,p_i)$ denotes a Bernoulli random variable that takes value $v_i$ with probability $p_i$ and $0$ otherwise.

The prophet’s expected value is
\[
\opt
= \mathbb{E}\!\left[\max\{X_1,X_2,X_3\}\right]
= 4p_3 - 2p_2(p_3-1) + (p_2-1)(p_3-1).
\]

We evaluate the performance of prediction-augmented threshold algorithms that post a base threshold equal to one of the available values $v_1=1$, $v_2=2$, or $v_3=4$.
Let $\mathrm{th}(i)$ denote the expected value obtained when the base threshold is set to $v_i$. A direct calculation yields
\[
\mathrm{th}(1) = 1- p_3 + 2\cdot p_2\cdot p_3 + 4\cdot(1-p_2)\cdot p_3, \qquad
\mathrm{th}(2) = 2p_2 + 4p_3(1-p_2), \qquad
\mathrm{th}(3) = 4p_3.
\]

We now examine how the optimal choice of threshold depends on the parameters $p_2$
and $p_3$.

\medskip
\noindent\textbf{Case 1: $p_3 = 0.2$.}
In this case, achieving a competitive ratio of at least $2/3$ requires:
\begin{itemize}
    \item choosing threshold $v_1=1$ when $p_2 \le \tfrac{2}{5}$, and
    \item choosing threshold $v_2=2$ when $p_2 > \tfrac{4}{7}$.
\end{itemize}

\medskip
\noindent\textbf{Case 2: $p_2 = 0.7$.}
Here, achieving a competitive ratio of at least $2/3$ requires:
\begin{itemize}
    \item choosing threshold $v_2=2$ when $p_3 \le \tfrac{17}{37}$, and
    \item choosing threshold $v_3=4$ when $p_3 > \tfrac{4}{5}$.
\end{itemize}

\medskip
This example demonstrates that the optimal base threshold depends delicately on the underlying distribution parameters, even when the prediction quality $\alpha$ is fixed.
This sensitivity motivates the need for a principled method for selecting thresholds, and highlights the challenge faced by algorithms that must perform well across a wide
range of instances.

\section{Proof of Lemma \ref{lem:transform}}

We present a sequence of transformations that converts an arbitrary
instance into a sorted scaled Bernoulli instance. We show that each transformation is
\emph{pessimistic} for threshold algorithms: it does not increase the algorithm’s expected reward (for any fixed base threshold, and for any prediction quality parameter), and thus cannot improve the competitive ratio. Consequently, the worst-case performance is attained within the class of sorted scaled Bernoulli instances.

\subsubsection*{Reduction from general to discrete distributions}

We show that it suffices to analyze instances consisting of discrete random
variables. Fix an instance
$I = (D_1,\dots,D_n)$, a base threshold $\tau$, and a prediction quality parameter
$\alpha \in [0,1]$. Let $\delta>0$ be arbitrary, and define
\[
\alpha' \coloneqq \frac{\alpha}{1+\delta}.
\]
We construct a discretized instance
$I' = (D'_1,\dots,D'_n)$ such that the performance of the prediction-augmented
threshold algorithm on $I'$ lower-bounds its performance on $I$, up to the
adjusted parameter~$\alpha'$.

\medskip
Let $X_i \sim D_i$ and define
\[
X_{\max} = \max_{1 \le i \le n} X_i, \qquad
M = \mathbb{E}[X_{\max}].
\]
Define the truncation level
\[
M' \coloneqq \inf\{v \ge 0 : \mathbb{E}[(X_{\max}-v)^+] \le \delta M\}.
\]
Since $\lim_{v\to\infty} \mathbb{E}[(X_{\max}-v)^+]=0$, the value $M'$ is well defined.

Next, define the discrete value set
\[
V \coloneqq \left\{ \delta(1+\delta)^k M :
0 \le k \le
\left\lceil \log_{1+\delta} \frac{M'}{\delta M} \right\rceil \right\}.
\]
For each $i\in[n]$, define a discretized random variable
\[
X_i' \coloneqq \min\{v \in V : v \ge \min(X_i,M')\},
\]
and let $D_i'$ denote its distribution. Thus, $X_i'$ is obtained by truncating
$X_i$ at $M'$ and rounding the result \emph{up} to the nearest value in $V$.

By construction, each $X_i'$ is a discrete nonnegative random variable.
Moreover, truncation ensures that the contribution of values above $M'$ is at
most $\delta M$, while rounding up preserves values within a multiplicative
factor of $1+\delta$. Hence,
\[
\mathbb{E}\!\left[\max_i X_i'\right]
\ge (1-\delta) M.
\]

We now compare the performance of prediction-augmented threshold algorithms on
$I$ and $I'$.

\begin{claim}
\label{cl:general_to_discrete}
For any realization $R$ of $I$ and the corresponding realization $R'$ of $I'$
obtained by applying the above rounding coordinatewise,
\[
\thr\!\left(R,\ \max\{\tau,\, \alpha \cdot \xm(R)\}\right)
\;\ge\;
\thr\!\left(R',\ \max\{\tau,\, \alpha' \cdot \xm(R')\}\right).
\]
Consequently,
\[
\thra(I,\tau) \ \ge\ \thr_{\alpha'}(I',\tau).
\]
\end{claim}

\begin{proof}
Fix a realization $R=(X_1,\dots,X_n)$ and let $R'=(X_1',\dots,X_n')$ be its
discretized counterpart. By construction,
\[
X_i' \le (1+\delta)X_i
\quad\text{and hence}\quad
\alpha \cdot \xm(R) \ge \alpha' \cdot \xm(R').
\]

We distinguish two cases.

\medskip
\noindent\textbf{Case A:} $\tau \ge \alpha \cdot \xm(R)$.  
In this case, the effective threshold in $R$ is $\tau$. Suppose that
\[
\thr\!\left(R',\max\{\tau,\alpha' \cdot \xm(R')\}\right)=X_i' > 0.
\]
Then $X_i' \ge \tau$ and $X_j' < \tau$ for all $j < i$. Since $\tau \in V$ and the
rounding operation is monotone, it follows that $X_i \ge \tau$ and
$X_j < \tau$ for all $j < i$ as well. Hence,
\[
\thr\!\left(R,\max\{\tau,\alpha \cdot \xm(R)\}\right) = X_i \ge X_i',
\]
as claimed.

\medskip
\noindent\textbf{Case B:} $\alpha \cdot \xm(R) \ge \tau$.  
Let
\[
\thr\!\left(R,\max\{\tau,\alpha\xm(R)\}\right)=X_i
\quad\text{and}\quad
\thr\!\left(R',\max\{\tau,\alpha'\xm(R')\}\right)=X_r'.
\]
Note that $i\neq\emptyset$ since the effective threshold is at most the realized
maximum.

If $i=r$, then $X_i \ge X_i'$ by construction.

If $i>r$, then $X_r' \leq X_r < \alpha \cdot \xm(R)$ while $X_i \ge \alpha \cdot \xm(R)$,
implying $X_i > X_r'$.

Finally, the case $i<r$ is impossible: if $X_i \ge \alpha \cdot \xm(R)$, then
\[
X_i' \ge \frac{X_i}{1+\delta}
\ge \frac{\alpha \cdot \xm(R)}{1+\delta}
= \alpha' \cdot \xm(R)
\ge \alpha' \cdot \xm(R'),
\]
and also $X_i' \ge \tau$. Hence $X_i'$ would have been selected before index~$r$,
a contradiction.

In all cases, the realized value under $R$ is at least that under $R'$.
\end{proof}

Since $\delta>0$ is arbitrary, we conclude that restricting attention to instances with discrete random variables incurs no loss of generality.

\subsubsection*{Discrete to Scaled}
Consider a general discrete instance $I = (X_1,\dots,X_n)$.
For each $i \in [n]$, suppose that $X_i$ is supported on a finite set
\[
\mathcal{D}_i = \{0, v_i^1, \dots, v_i^k\},
\qquad
0 < v_i^1 < \cdots < v_i^k.
\]
We construct $k$ new random variables $X_i^1,\dots,X_i^k$, where each $X_i^j$ is a Bernoulli
variable supported on $\{0, v_i^j\}$ with probability
\[
\Pr[X_i^j = v_i^j]
    = \frac{\Pr[X_i = v_i^j]}{\Pr[X_i \le v_i^j]}.
\]
By construction, the maximum of the set $\{X_i^1,\dots,X_i^k\}$ has the same distribution as $X_i$,
that is,
\[
X_i^{\max} := \max_{1\le j\le k} X_i^j \ \overset{d}{=} \ X_i.
\]
Let $\tilde{I}$ denote the expanded instance in which each variable $X_i$ is replaced by the
collection $X_i^1,\dots,X_i^k$.

\begin{claim}
For any base threshold $\tau$ and any prediction quality $\alpha$,  
\[
\thra(I,\tau) \ \ge \ \thra(\tilde{I}, \tau).
\]
\end{claim}
\begin{proof}
Fix a realization $R = (v_1,\dots,v_n)$ of $I$, where $v_i \in \mathcal{D}_i$, and suppose
$v_i = v_i^j$ for some $j$.  
We couple this realization with a corresponding realization in the expanded instance $\tilde{I}$ as follows:
all variables $X_\ell$ with $\ell \neq i$ take the same value as in $R$, and among the Bernoulli copies
$X_i^1,\dots,X_i^k$, exactly one copy (namely $X_i^j$) takes the value $v_i^j$, while all higher-indexed copies
take the value $0$.  
This ensures that the realized maximum at coordinate~$i$ is the same in both instances, and hence
$\xm(R)$ is preserved.

Under prediction quality $\alpha$, the prediction-augmented threshold used by the algorithm is
\[
\zeta = \max\{\tau,\ \alpha \cdot \xm(R)\},
\]
and therefore both instances use the same effective threshold~$\zeta$.

If either $\zeta > v_i^j$ or $\zeta \le v_i^j$ but some $\ell < i$ satisfies $v_\ell \ge \zeta$,
then the algorithm selects a coordinate strictly before $i$ in both the original and expanded instances.
Thus the chosen index $i^*$ is the same, and the output values coincide.

Otherwise, we must have $\zeta \le v_i^j$ and $\zeta > \max_{\ell < i} v_\ell$, meaning that
coordinate~$i$ is the first to exceed the threshold.  
In the original instance $I$, the algorithm outputs $v_i^j$.  
In the expanded instance $\tilde{I}$, the algorithm outputs at most $v_i^j$, since among the Bernoulli
copies only $X_i^j$ equals $v_i^j$ while all others are~0.

Thus in every coupled realization, the algorithm’s output under $\tilde{I}$ is never larger than its
output under $I$.  
Taking expectations over realizations yields
\[
\thra(I,\tau) \ \ge \ \thra(\tilde{I}, \tau),
\]
completing the proof.
\end{proof}

Second, we prove that for an instance $I$ of scaled Bernoulli variables, we may assume without
loss of generality that it is sorted.

Formally, recall that in such an instance each $X_i$ is supported on $\{0,v_i\}$.
Let $I = (X_1,\dots,X_n)$ be any such instance, and let $\pi$ be a permutation of $[n]$ such that
\[
v_{\pi(1)} \le v_{\pi(2)} \le \cdots \le v_{\pi(n)}.
\]
Define the sorted instance
\[
I^{\mathrm{sort}} = (Y_1,\dots,Y_n) \quad\text{where } Y_k := X_{\pi(k)}.
\]

\begin{claim}
For any base threshold $\tau$ and prediction quality $\alpha$,
\[
\thra(I,\tau) \ \ge\ \thra(I^{\mathrm{sort}},\tau).
\]
Thus, it suffices to analyze sorted scaled Bernoulli instances.
\end{claim}

\begin{proof}
It is enough to show that swapping any adjacent inverted pair cannot increase the
algorithm’s expected value.  
Consider two instances that differ only in swapping $X_i$ and $X_{i+1}$, where
$v_i > v_{i+1}$.  
Couple the two instances so that every $X_\ell$ has the same realized value in both;
only their order differs.  
Since the multiset of realized values is unchanged, the prophet value $\xm(R)$ and the
effective threshold
\[
\zeta = \max\{\tau,\, \alpha \cdot \xm(R)\}
\]
are identical in both instances.

Up to position $i-1$, both instances behave identically.  
If neither or only one of the two values exceeds $\zeta$, both instances either skip or select
the same value, so the outputs match.  
If both exceed $\zeta$, then in the original order the larger value $v_i$ appears first and is
selected, while in the swapped order the smaller value $v_{i+1}$ appears first and is selected.
Thus the swapped instance never yields a larger output.

Taking expectations over the coupling, we obtain
\[
\thra(I,\tau) \ge \thra(I',\tau),
\]
where $I'$ is the instance with the pair swapped.  
Successively eliminating all inversions yields the sorted instance $I^{\mathrm{sort}}$ and the claim.
\end{proof}

The lemma follows immediately by combining the three preceding claims: the first reduction shows that we may assume, without loss of generality, that each random variable is discrete the second reduction
shows that we may assume, without loss of generality, that each random variable is supported on $\{0,v_i\}$, and the second reduction shows that we may further assume the values are sorted in nondecreasing order. Therefore, it suffices to analyze sorted scaled Bernoulli instances.

\section{Deferred Proofs}\label{app:deferred_proofs}

\begin{proof}[Proof of Lemma~\ref{lem:known_alpha_lb}]
Consider an instance with two random variables. Let $X_1 \equiv 1$ be deterministic, and let
\[
X_2 =
\begin{cases}
\frac{1}{\alpha}, & \text{with probability } \alpha,\\
0, & \text{otherwise}.
\end{cases}
\]
The prophet obtains
\[
\opt
= \mathbb{E}[\max\{X_1,X_2\}]
= \alpha\cdot \frac{1}{\alpha} + (1-\alpha)\cdot 1
= 2-\alpha.
\]

Now consider the (fixed) prediction $\tilde v \equiv 1$. Since $\max\{X_1,X_2\} \ge 1$ always
and, in this instance, $\max\{X_1,X_2\} \le \frac{1}{\alpha}$, we have
\[
\tilde v = 1 \;\ge\; \alpha \cdot \max\{X_1,X_2\}.
\]

Under this prediction, no online algorithm can obtain an expected value exceeding~$1$.
Indeed, upon observing $X_1 = 1$, the algorithm either accepts and receives~$1$, or rejects
and proceeds to $X_2$, whose expectation is
\[
\mathbb{E}[X_2]
= \alpha\cdot \frac{1}{\alpha} + (1-\alpha)\cdot 0
= 1.
\]
Consequently, the best achievable competitive ratio is at most
\[
\frac{1}{\opt} = \frac{1}{2-\alpha},
\]
as claimed.
\end{proof}

\begin{proof}[Proof of Lemma \ref{lem:g(p) func}]
To analyze the convexity, we compute the second derivative of $g$ with respect to $p$. We begin by simplifying the expression using the substitution $u = 1-p$. Note that since $p \in (0,1)$, we have $u \in (0,1)$.

Substituting $p = 1-u$ into the numerator:
\begin{align*}
4p^2 - 6p + 3 &= 4(1-u)^2 - 6(1-u) + 3 \\
&= 4(1 - 2u + u^2) - 6 + 6u + 3 \\
&= 4u^2 - 2u + 1.
\end{align*}
The function $g(p)$ can thus be rewritten more simply in terms of $u$:
\[
g(p) = \frac{4u^2 - 2u + 1}{4u} = u - \frac{1}{2} + \frac{1}{4}u^{-1}.
\]
Differentiating twice using the chain rule and substituting back $u = 1-p$, we obtain:
\[
g''(p) = \frac{1}{2(1-p)^3}.
\]
Hence  $g''(p) > 0$ for $p \in (0,1)$. Therefore, $g$ is strictly convex on $(0,1)$.

To prove the bound, we observe that the strictly convex function $g$ attains its global minimum at the critical point where $g'(p) = 0$.
Setting $g'(p) = 0$:
\[
\frac{1}{4(1-p)^2} - 1 = 0 \implies (1-p)^2 = \frac{1}{4} \implies p = \frac{1}{2}.
\]
Evaluating the function at this minimum:
\[
g\left(\frac{1}{2}\right) = \frac{4(1/4) - 6(1/2) + 3}{4(1/2)} = \frac{1}{2}.
\]
Thus, $g(p) \ge \frac{1}{2}$ for all $p \in (0,1)$.
\end{proof}

\begin{proof}[Proof of Proposition \ref{prop:reduced_optimization}]
Let $A(x_1, x_2) = k x_2 + x_1(1-x_2)$ and $B(x_0, x_1, x_2) = x_0(1-x_1)(1-x_2)$. The constraint is $A = tB$. The denominator of $f$ is $D = A + B$. Substituting $B = A/t$, we have $D = A(1 + 1/t) = \frac{t+1}{t}A$.

The numerator of $f$ is $N = x_1 + k x_2(1-x_1) = x_1 + k x_2 - k x_1 x_2$. We can relate this to $A$ as $N = A - (k-1)x_1 x_2$.
The objective function simplifies to
\[ f = \frac{N}{D} = \frac{A - (k-1)x_1 x_2}{\frac{t+1}{t}A} = \frac{t}{t+1} \left( 1 - (k-1) \frac{x_1 x_2}{A} \right). \]
Let $g(x_1, x_2) = \frac{x_1 x_2}{A} = \frac{x_1 x_2}{x_1 + k x_2 - x_1 x_2}$. Then $f = \frac{t}{t+1} (1 - (k-1)g)$.
The feasible region $\mathcal{F}$ for $(x_1, x_2)$ is defined by $x_1, x_2 \in [0,1]$ and the implicit constraint $x_0 \le 1$, which translates to $A \le t(1-x_1)(1-x_2)$.

Since $k>1$, minimizing $f$ is equivalent to maximizing $g(x_1, x_2)$ over $\mathcal{F}$. We compute the partial derivatives of $g$ using the quotient rule.

For $x_1$:

\begin{align*}
\frac{\partial g}{\partial x_1} 
= & \frac{(x_2)(x_1 + k x_2 - x_1 x_2) - (x_1 x_2)(1 - x_2)}{A^2}\\
= &  \frac{x_1 x_2 + k x_2^2 - x_1 x_2^2 - x_1 x_2 + x_1 x_2^2}{A^2} = \frac{k x_2^2}{A^2} \ge 0.
\end{align*}

For $x_2$:
\begin{align*}
  \frac{\partial g}{\partial x_2}
  = & \frac{(x_1)(x_1 + k x_2 - x_1 x_2) - (x_1 x_2)(k - x_1)}{A^2} \\
  = & \frac{x_1^2 + k x_1 x_2 - x_1^2 x_2 - k x_1 x_2 + x_1^2 x_2}{A^2} = \frac{x_1^2}{A^2} \ge 0.
\end{align*}

As $g$ is continuous and non-decreasing on the compact set $\mathcal{F}$, its maximum must be attained on the upper boundary $A = t(1-x_1)(1-x_2)$, which corresponds to $x_0 = 1$.
This is equivalent to minimizing the reciprocal $h(x_1, x_2) = 1/g = \frac{1}{x_2} + \frac{k}{x_1} - 1$ subject to the same constraint.
We minimize $\bar{h}= \frac{1}{x_2} + \frac{k}{x_1}$ subject to the constraint, which simplifies to
\[ (t+1)x_1 + (k+t)x_2 - (t+1)x_1x_2 = t. \]
We can solve for $x_1$ and substitute into $\bar{h}$ to obtain a function of $x_2$:
\[ x_1 = \frac{t - (k+t)x_2}{(t+1)(1-x_2)} \quad \implies \quad \bar{h}(x_2) = \frac{1}{x_2} + \frac{k(t+1)(1-x_2)}{t - (k+t)x_2}. \]
The derivative is
\[ \bar{h}'(x_2) = -\frac{1}{x_2^2} + \frac{k^2(t+1)}{(t - (k+t)x_2)^2}. \]
Setting $\bar{h}'(x_2) = 0$ and taking the positive root yields the optimal point
\[ x_2^* = \frac{t}{k+t + k\sqrt{t+1}}. \]
At this point, the maximum value of $g$ is found by minimizing $\bar{h}$. The minimum value of $\bar{h}-1$ is $1/g_{max}$:
\[ \bar{h}(x_2^*) - 1 = \frac{k(2+t+2\sqrt{t+1})}{t} \implies g_{max} = \frac{t}{k(2+t+2\sqrt{t+1})}. \]
Finally, we substitute $g_{max}$ back into the expression for $f$:
\begin{align*}
 f_{min}  =  \frac{t}{t+1} \left( 1 - (k-1)g_{max} \right) & = \frac{t}{t+1} \left( 1 - \frac{t(k-1)}{k(2+t+2\sqrt{t+1})} \right) \\
   =  \frac{t}{t+1} \left( \frac{k(2+t+2\sqrt{t+1}) - t(k-1)}{k(2+t+2\sqrt{t+1})} \right) &= \frac{t}{t+1} \left( \frac{2k + t + 2k\sqrt{t+1}}{k(2+t+2\sqrt{t+1})} \right). 
\end{align*}
This completes the proof.
\end{proof}

\begin{proof}[Proof of \Cref{lem:tight_oblivious_instance}]
For $\alpha\in (0,1)$, consider the following instance, where we assume that $\varepsilon$ is chosen sufficiently small such that $1+ \varepsilon<1/\alpha$.
\begin{align*}
X_0 &= 1 \quad \text{w.p}\quad \frac{1+\varepsilon}{1+2\varepsilon}\\
X_1 &= 1+\varepsilon \quad \text{w.p}\quad \frac{1}{2(1+\varepsilon)}\\ 
X_2 &= 1/\alpha \quad \text{w.p} \quad \frac{\alpha}{\alpha+1}
\end{align*}

To compute $\mathbb{E}[M \cdot \mathbbm{1}[M \ge 1+\varepsilon]]$, we observe that the condition $M \ge 1+\varepsilon$ is satisfied if $M = 1/\alpha$ or $M = 1+\varepsilon$. Since $1/\alpha$ is the unique maximum value possible, $M=1/\alpha$ occurs precisely when $X_2 = 1/\alpha$, regardless of the values of $X_1$ and $X_0$. The value $M=1+\varepsilon$ occurs when $X_1 = 1+\varepsilon$ and $X_2 = 0$. Summing these contributions yields:

\begin{align*}
\mathbb{E}[M \cdot \mathbbm{1}[M \ge 1+\varepsilon]] &= \frac{1}{\alpha} \cdot {\Pr}\left[X_2 = \frac{1}{\alpha}\right] + (1+\varepsilon) \cdot {\Pr}[X_1 = 1+\varepsilon, X_2 = 0] \\
&= \frac{1}{\alpha} \left( \frac{\alpha}{\alpha+1} \right) + (1+\varepsilon) \left( \frac{1}{2(1+\varepsilon)} \right) \left( \frac{1}{\alpha+1} \right) \\
&= \frac{1}{\alpha+1} + \frac{1}{2(\alpha+1)} \\
&= \frac{3}{2(\alpha+1)}.
\end{align*}

Next, we compute $\mathbb{E}[M \cdot \mathbbm{1}[M < 1+\varepsilon]]$. The condition $M < 1+\varepsilon$ implies that $M$ cannot take the values $1/\alpha$ or $1+\varepsilon$. The contribution is:

\begin{align*}
\mathbb{E}[M \cdot \mathbbm{1}[M < 1+\varepsilon]] &= 1 \cdot {\Pr}[X_0=1, X_1=0, X_2=0] \\
&= \left( \frac{1+\varepsilon}{1+2\varepsilon} \right) \left( 1 - \frac{1}{2(1+\varepsilon)} \right) \left( \frac{1}{\alpha+1} \right) \\
&= \left( \frac{1+\varepsilon}{1+2\varepsilon} \right) \left( \frac{1+2\varepsilon}{2(1+\varepsilon)} \right) \left( \frac{1}{\alpha+1} \right)\\
&= \frac{1}{2(\alpha+1)}.
\end{align*}

Hence, even though distributions are not continuous, $T_X(3/4)$ is well defined. In this case we obtain 
$T_X(3/4) = 1+\varepsilon$.
We also compute
\begin{equation*}
\talga(X,T_X(3/4)) = \frac{1}{2} + \frac{1+2\varepsilon}{2(\alpha+1)(1+\varepsilon)}.
\end{equation*} to obtain
\begin{align*}
\frac{\talga(X,T_X(3/4)) }{\mathbb{E}[M]} &= \frac{\alpha}{4} + \frac{1}{4} + \frac{(1+\varepsilon) + \varepsilon}{4(1+\varepsilon)} \\
&= \frac{1}{2} + \frac{\alpha}{4} + \frac{\varepsilon}{4(1+\varepsilon)},
\end{align*}
proving the claim.

For $\alpha = 1$, consider the instance $Y = (Y_0,Y_1)$ where $Y_0 = 1$ w.p $2/3$ and $Y_1 = 2$ w.p $1/2$. Then $\E[M\mathbbm{1}[M \geq 2]] = 1$ and $\E[M\mathbbm{1}[M < 2]] = (1/2)\cdot (2/3)=1/3$. Thus $T_Y(3/4) = 2$ and so $\talga(Y,T_Y(3/4)) = 1$ while $\E[M] = 4/3$, as desired.

Finally, recall that $\alpha = 0$ is just the standard prophet inequality. 
\end{proof}

\begin{proof}[Proof of Lemma \ref{lem:smallval fn decreasing}]
We introduce the substitution $u = \sqrt{1-c} \Rightarrow c = 1 - u^2$. 
We analyze the function's monotonicity with respect to $u$ on the domain $u \in [0, 1]$.

By the chain rule, $\frac{dg}{dc} = \frac{dg}{du} \cdot \frac{du}{dc}$. First, we compute
\[ \frac{du}{dc} = \frac{d}{dc}(1-c)^{1/2} = -\frac{1}{2\sqrt{1-c}} = -\frac{1}{2u}. \]
Since $u > 0$ for $c \in [0, 1)$, $\frac{du}{dc}$ is strictly negative on $[0, 1)$. We must show that $g(u)$ is strictly increasing.

We rewrite $g$ as a function of $u$ 
\[ g(u) = \frac{(2 - \alpha)u^2 + 2u + \alpha}{(u+1)^2}. \]
We differentiate $g(u)$ using the quotient rule:
\begin{align*}
    \frac{dg}{du} &= \frac{[2(2 - \alpha)u + 2](u+1)^2 - [(2 - \alpha)u^2 + 2u + \alpha][2(u+1)]}{(u+1)^4} \\
    &= \frac{(4 - 2\alpha)u^2 + (6 - 2\alpha)u + 2 - (4 - 2\alpha)u^2 - 4u - 2\alpha}{(u+1)^3} \\
    &= \frac{(2 - 2\alpha)u + (2 - 2\alpha)}{(u+1)^3}\\
    &= \frac{2(1 - \alpha)}{(u+1)^2}.
\end{align*}
Since $\alpha \in (0, 1)$, the numerator $2(1 - \alpha)$ is a positive constant. The denominator $(u+1)^2$ is also positive for all $u \in [0, 1]$.
Therefore, $\frac{dg}{du}$ is strictly positive on $[0, 1]$.

As $\frac{dg}{du} > 0$ on $[0, 1]$ and $\frac{du}{dc} < 0$ on $[0, 1)$, it follows that $\frac{dg}{dc} = \frac{dg}{du} \cdot \frac{du}{dc} < 0$ for all $c \in [0, 1)$.
Since $g(c)$ is continuous on the closed interval $[0, 1]$ and strictly decreasing on $[0, 1)$, it is strictly decreasing on the entire interval $[0, 1]$.
\end{proof}

\begin{proof}[Proof of \Cref{lem:T/alpha probability upper bound}]
By Markov's inequality,
\[
    \E[M(X)| M(X) \geq \bar{\tau}]\Pr[M(X) \geq \bar{\tau}] \geq \frac{\bar{\tau}}{\alpha} \cdot \Pr[M(X)\geq \bar{\tau}/\alpha]
\]
We may further upper bound the left hand side of the equation and apply the definition of $\bar{\tau}$.
\begin{align*}
& \E[M(X)| M(X) \geq \bar{\tau}]\Pr[M(X) \geq \bar{\tau}] \\
= & 3\E[M(X)| M(X) < \bar{\tau}]\Pr[M(X) < \bar{\tau}] \\
\leq &  3\cdot \bar{\tau} \cdot (1-\Pr[M(X) \geq \bar{\tau}])   
\end{align*}
Combining these equations we obtain
\begin{align*}
& 3\cdot \bar{\tau} \cdot (1-\Pr[M(X) \geq \bar{\tau}]) \geq \frac{\bar{\tau}}{\alpha} \cdot \Pr[M(X)\geq \bar{\tau}/\alpha]\\
\Rightarrow \quad & 
3 \geq \frac{1}{\alpha} \cdot \Pr[M(X)\geq \bar{\tau}/\alpha] + 3\cdot \Pr[M(X) \geq \bar{\tau}] \geq (1/\alpha + 3) \Pr[M(X)\geq \bar{\tau}/\alpha]\\
\Rightarrow \quad & \frac{1}{1/(3\alpha)+1} \geq  \Pr[M(X)\geq \bar{\tau}/\alpha].
\end{align*}
\end{proof}

\begin{proof}[Proof of \Cref{lem:order_stats_bound}]
Let $p_i = \Pr[X_i > L_1]$ and $r_i = \Pr[X_i > L_2]$. Since $L_1 > L_2$, we have $\{X_i > L_1\} \subseteq \{X_i > L_2\}$, which implies $p_i \leq r_i$ for all $i$.
Define the events $A = \{M_1 > L_1\}$, $B = \{M_2 > L_2\}$, and $C = \{M_1 > L_2\}$. The inequality to be shown is $\Pr[B \mid A] \leq \Pr[C]$.
Using the identity $\Pr[B \mid A] = 1 - \frac{\Pr[A \cap B^c]}{\Pr[A]}$, the target inequality is equivalent to:
\[
1 - \Pr(C) \leq \frac{\Pr[A \cap B^c]}{\Pr[A]} \iff \Pr[C^c] \Pr[A] \leq \Pr[A \cap B^c].
\]
We evaluate the terms as follows:
\begin{align*}
\Pr[C^c] &= \prod_{i=1}^n (1 - r_i), \\
\Pr[A] &= 1 - \prod_{i=1}^n (1 - p_i).
\end{align*}
The event $A \cap B^c$ corresponds to the case where at least one variable exceeds $L_1$, but at most one variable exceeds $L_2$. Since $L_1 > L_2$, any variable exceeding $L_1$ also exceeds $L_2$. Thus, $A \cap B^c$ occurs if and only if exactly one variable exceeds $L_2$, and that same variable also exceeds $L_1$. Therefore:
\[
\Pr[A \cap B^c] = \sum_{k=1}^n \left( p_k \prod_{j \neq k} (1 - r_j) \right) = \left( \prod_{i=1}^n (1-r_i) \right) \sum_{k=1}^n \frac{p_k}{1-r_k}.
\]
Substituting these into the required inequality and dividing by the common positive factor $\prod_{i=1}^n (1-r_i)$, we must show:
\[
1 - \prod_{i=1}^n (1 - p_i) \leq \sum_{k=1}^n \frac{p_k}{1-r_k}.
\]
By the Union Bound, we know that $1 - \prod_{i=1}^n (1 - p_i) \leq \sum_{k=1}^n p_k$.
Furthermore, we have $1-r_k \le 1$, which implies $p_k \leq \frac{p_k}{1-r_k}$.
Combining these bounds yields:
\[
1 - \prod_{i=1}^n (1 - p_i) \leq \sum_{k=1}^n p_k \leq \sum_{k=1}^n \frac{p_k}{1-r_k}.
\]
Thus the inequality holds.
\end{proof}

\begin{proof}[Proof of Lemma \ref{lem:compare_ratios}]
Let $L(x)$ be the left-hand side of the inequality. We can substitute $g(x)$ to simplify $L(x)$:
\[ L(x) = x \left( \frac{3x^2}{1+3x^2} \right) + \left( 1 - \frac{3x^2}{1+3x^2} \right) = \frac{3x^3+1}{1+3x^2} \]
We wish to prove that $L(x) \ge \frac{2+x^3}{3}$ for $x \in [0,1]$. Let $h(x)$ be the difference:
\[ h(x) = L(x) - \frac{2+x^3}{3} = \frac{3x^3+1}{1+3x^2} - \frac{2+x^3}{3} \]
To show $h(x) \ge 0$, we place the terms over a common denominator:
\[ h(x) = \frac{3(3x^3+1) - (2+x^3)(1+3x^2)}{3(1+3x^2)} \]
The denominator $3(1+3x^2)$ is strictly positive. Thus, we only need to prove that the numerator, $N(x) = 3(3x^3+1) - (2+x^3)(1+3x^2)$, is non-negative on $[0,1]$.
Expanding $N(x)$, we have:
\[ N(x) = (9x^3 + 3) - (2 + 6x^2 + x^3 + 3x^5) = -3x^5 + 8x^3 - 6x^2 + 1 \]
We analyze $N(x)$ on the interval $[0,1]$. First, we check the boundary values:
\[ N(0) = 1 > 0, \quad N(1) = 0. \]
Next, we find the derivative of $N(x)$ to determine its monotonicity:
\[ N'(x) = -15x^4 + 24x^2 - 12x = -3x(5x^3 - 8x + 4) \]
Let $p(x) = 5x^3 - 8x + 4$. We analyze the sign of $p(x)$ on $[0,1]$.
The derivative is $p'(x) = 15x^2 - 8$, which has a single root in $[0,1]$ at $\sqrt{8/15}$. 
The value at this minimum is:
\[ p(\sqrt{8/15}) = 4 - \frac{16}{3}\sqrt{\frac{8}{15}}>0. \]
Since $p(0)=4$, $p(1)=1$, and its minimum on the interval is positive, $p(x) > 0$ for all $x \in [0,1]$.

Given that $p(x) > 0$ and $-3x \le 0$ for $x \in [0,1]$, $N'(x) = -3x \cdot p(x) \leq 0$ on $[0,1]$. This implies that $N(x)$ is a non-increasing function on $[0,1]$.
Because $N(x)$ is non-increasing, its minimum value on $[0,1]$ must be at $x=1$. We conclude that $h(x) \ge 0$ for all $x \in [0,1]$.
\end{proof}

\paragraph{\textbf{Point Masses}}\label{sec:extension to general}

When distributions of variables $X = (X_1,\ldots,X_n)$ are not continuous, it may be that there does not exist any deterministic threshold satisfying the definition of $\bar{\tau}$. However, it is easy to see that the following will hold instead. Let $\varepsilon \sim U[0,\delta]$ for some value $\delta$. Then for any $c\in (0,1)$ there exists $\delta$ sufficiently small defining a value $\tau$ such that

\[
\E[M \ | \ M \geq \tau+\varepsilon]\Pr[ M \geq \tau+\varepsilon] = c \cdot \opt,
\]
where the expectation is also taken over the randomness of $\varepsilon\sim U[0,\delta]$. Now with $c=3/4$, consider the algorithm that is posting the maximum of $\alpha\cdot M$ and the randomized base threshold $\tau' = \tau + \varepsilon$. Clearly, Lemma \ref{lem:>T/alpha values} for values that use the prediction is not affected by this change in the algorithm. For Lemma \ref{lem:T/alpha probability upper bound}, we only need
\begin{align*}
\E[M(X)| M(X) \geq {\tau'}]\Pr[M(X) \geq {\tau'}] = 3\E[M(X)| M(X) < {\tau'}]\Pr[M(X) < {\tau'}]
\end{align*}
to hold, which is true by definition of $\tau'$. The other part is to check that the new randomized algorithm obtains the same welfare in the case of the base threshold, and not the prediction being used. We observe that its performance is the same up to an additive $\delta$ of the original threshold algorithm $\max\{\bar{\tau},M\cdot\alpha\}$ being applied to the continuous instance $Y = (Y_1,\ldots,Y_n)$ where $Y_i$ is a smoothing of $X_i$ by some noise $\delta \cdot U[0,1]$.

\end{document}